\documentclass[10pt, conference, compsocconf]{IEEEtran}
\IEEEoverridecommandlockouts

\usepackage{cite}
\usepackage[pdftex]{graphicx}
\usepackage[cmex10]{amsmath}
\usepackage[caption=false,font=footnotesize]{subfig}
\usepackage{url}

\usepackage{array}
\usepackage{microtype}
\usepackage{amsthm}
\usepackage{amssymb}
\usepackage{mathtools}
\usepackage{verbatim}
\usepackage{color}
\usepackage{boxedminipage}
\usepackage{slashbox}

\let\labelindent\relax
\usepackage{enumitem}

\usepackage{hyperref}

\newcommand{\ie}{i.e.}
\newcommand{\etc}{etc.}
\newcommand{\eg}{e.g.}
\newcommand{\etal}{et al.}

\mathchardef\mhyphen="2D
\newcommand{\Hyphen}{\mhyphen}

\newtheorem{theorem}{Theorem}

\newcommand{\mycomment}[1]{\emph{\textcolor{red}{[#1]}}}

\newcommand{\RISCV}{RISC-V}

\newcommand{\regfile}{atomic}

\newcommand{\LdInst}{\mathsf{Ld}}
\newcommand{\StInst}{\mathsf{St}}

\newcommand{\ComInst}{\mathsf{Commit}}
\newcommand{\RecInst}{\mathsf{Reconcile}}

\newcommand{\LWComInst}{\mathsf{LWCom}}
\newcommand{\LdLdFence}{\mathsf{FENCE_{LL}}}

\newcommand{\syncInst}{\mathsf{sync}}
\newcommand{\lwsyncInst}{\mathsf{lwsync}}
\newcommand{\isyncInst}{\mathsf{isync}}
\newcommand{\cmpInst}{\mathsf{cmp}}
\newcommand{\bcInst}{\mathsf{bc}}

\newcommand{\MBInst}{\mathsf{MEMBAR}}

\newcommand{\wmmSSB}{WMM-S}

\newcommand{\IIE}{\ensuremath{\mathrm{I^2E}}}

\newcommand{\True}{\mathsf{True}}
\newcommand{\False}{\mathsf{False}}

\newcommand{\coOrd}{<_{co}}

\newcommand{\tsoNmRule}{TSO-Nm}
\newcommand{\tsoLdRule}{TSO-Ld}
\newcommand{\tsoStRule}{TSO-St}
\newcommand{\tsoComRule}{TSO-Com}
\newcommand{\tsoDeqSbRule}{TSO-DeqSb}

\newcommand{\psoDeqSbRule}{PSO-DeqSb}

\newcommand{\wmmNmRule}{WMM-Nm}
\newcommand{\wmmLdRule}{WMM-Ld}

\newcommand{\wmmStRule}{WMM-St}
\newcommand{\wmmRecRule}{WMM-Rec}
\newcommand{\wmmComRule}{WMM-Com}
\newcommand{\wmmDeqSbRule}{WMM-DeqSb}

\newcommand{\wmmSSBPropSt}{\wmmSSB-Copy}

\newcommand{\wmmSSBDeqSbRule}{\wmmSSB-DeqSb}

\newcommand{\ccmModelName}{CCM}
\newcommand{\flowModelName}{FM} 

\newcommand{\flowReorderRule}{\flowModelName-Reorder}
\newcommand{\flowFlowRule}{\flowModelName-Flow}
\newcommand{\flowBypassRule}{\flowModelName-Bypass}

\newcommand{\ProgOrd}{<_{po}}
\newcommand{\MemOrd}{<_{mo}}

\newcommand{\ReadFrom}{\xrightarrow{r\!f}}

\newcommand{\AxiInstOrd}{Inst-Order}
\newcommand{\AxiLdVal}{Ld-Val}

\newcommand{\orderFunc}{\mathsf{order}}

\setlist{noitemsep, leftmargin=*,topsep=2pt}

\begin{document}
\bstctlcite{bstctl:nodash}

\title{Weak Memory Models: Balancing Definitional Simplicity and Implementation Flexibility}


\author{\IEEEauthorblockN{Sizhuo Zhang, Muralidaran Vijayaraghavan, Arvind}
\IEEEauthorblockA{Computer Science and Artificial Intelligence Laboratory\\
Massachusetts Institute of Technology\\
\{szzhang, vmurali, arvind\}@csail.mit.edu}\thanks{A version of this paper appears in the 26th International Conference on Parallel Architectures and Compilation Techniques (PACT), September, 2017. DOI: 10.1109/PACT.2017.29. \copyright 2017 IEEE.}}

\maketitle

\begin{abstract}
The memory model for \RISCV, a newly developed open source ISA, has not been finalized yet and thus, offers an opportunity to evaluate existing memory models.
We believe \RISCV{} should not adopt the memory models of POWER or ARM, because their axiomatic and operational definitions are too complicated.
We propose two new weak memory models: WMM and \wmmSSB{}, which balance definitional simplicity and implementation flexibility differently.
Both allow all instruction reorderings except overtaking of loads by a store. 
We show that this restriction has little impact on performance and it considerably simplifies operational definitions. 
It also rules out the out-of-thin-air problem that plagues many definitions. 
WMM is simple (it is similar to the Alpha memory model), but it disallows behaviors arising due to shared store buffers and shared write-through caches (which are seen in POWER processors).
\wmmSSB{}, on the other hand, is more complex and allows these behaviors.
We give the operational definitions of both models using \emph{Instantaneous Instruction Execution} (I\textsuperscript{2}E), which has been used in the definitions of SC and TSO.
We also show how both models can be implemented using conventional cache-coherent memory systems and out-of-order processors, and encompasses the behaviors of most known optimizations.

\end{abstract}


%
\IEEEpeerreviewmaketitle

\section{Introduction}\label{sec: intro}



\newcolumntype{P}[1]{>{\raggedright\arraybackslash}p{#1}}
\begin{figure*}[t]
    \centering
    \footnotesize
    \begin{tabular}{|p{9ex}|p{14.5ex}|p{20ex}|p{17.5ex}|p{19ex}|p{24ex}|p{14ex}|}
        \hline
        & \multicolumn{2}{|c|}{Definition} & \multicolumn{4}{|c|}{Model properties / Implementation flexibility} \\ \cline{2-7}
        & Operational model & Axiomatic model & Store atomicity & Allow shared write-through cache/shared store buffer & Instruction reorderings & Ordering of data-dependent loads \\ \hline
        SC & Simple; \IIE{}~\cite{lamport1979make} & Simple~\cite{lamport1979make} & Single-copy atomic & No & None & Yes \\ \hline
        TSO & Simple; \IIE{}~\cite{sewell2010x86} & Simple~\cite{sparc1992sparcv8} & Multi-copy atomic & No & Only St-Ld reordering & Yes \\ \hline
        RMO & Doesn't exist & Simple; needs fix~\cite{weaver1994sparc} & Multi-copy atomic & No & All four & Yes \\ \hline
        Alpha & Doesn't exist & Medium~\cite{alpha1998} & Multi-copy atomic & No & All four & No \\ \hline
        RC & Doesn't exist & Medium~\cite{gharachorloo1990memory} & Unclear & \emph{\textcolor{red}{No}} & All four & Yes \\ \hline
        ARM and POWER & Complex; non \IIE{}~\cite{sarkar2011understanding,flur2016modelling} & Complex \cite{mador2012axiomatic,alglave2014herding} & Non-atomic & Yes & All four & Yes \\ \hline
        WMM & Simple; \IIE{} & Simple & Multi-copy atomic & No & All except Ld-St reordering & No \\ \hline
        \wmmSSB{} & Medium; \IIE{} & Doesn't exist & Non-atomic & Yes & All except Ld-St reordering & No \\ \hline
    \end{tabular}
    \caption{Summary of different memory models}\label{tab: mm sumary}
\end{figure*}

A memory model for an ISA is the specification of all legal multithreaded program behaviors.
If microarchitectural changes conform to the memory model, software remains compatible.
Leaving the meanings of corner cases to be implementation dependent makes the task of proving the correctness of multithreaded programs, microarchitectures and cache protocols untenable.
While strong memory models like SC and SPARC/Intel-TSO are well understood, weak memory models of commercial ISAs like ARM and POWER are driven too much by microarchitectural details, and inadequately documented by manufacturers.
For example, the memory model in the POWER ISA manual~\cite{power2013version} is ``defined'' as reorderings of events, and an event refers to performing an instruction with respect to a processor. 
While reorderings capture some properties of memory models, it does not specify the result of each load, which is the most important information to understand program behaviors.
This forces the researchers to formalize these weak memory models by empirically determining the allowed/disallowed behaviors of commercial processors and then constructing models to fit these observations~\cite{alglave2009semantics,Alglave2011,alglave2012formal,mador2012axiomatic,alglave2014herding,sarkar2011understanding,sarkar2012synchronising,alglave2013software,flur2016modelling}.

The newly designed open-source \RISCV{} ISA~\cite{riscv} offers a unique opportunity to reverse this trend by giving a clear definition with understandable implications for implementations.
The \RISCV{} ISA manual only states that its memory model is weak in the sense that it allows a variety of instruction reorderings~\cite{waterman2016riscv}.
However, so far no detailed definition has been provided, and the memory model is not fixed yet.

In this paper we propose two weak memory models for \RISCV{}: WMM and \wmmSSB{}, which balance definitional simplicity and implementation flexibility differently.
The difference between the two models is regarding store atomicity, which is often classified into the following three types~\cite{lustig2015armor}:
\begin{itemize}
    \item \emph{Single-copy atomic}: a store becomes visible to all processors at the same time, \eg{}, in SC.
    \item \emph{Multi-copy atomic}: a store becomes visible to the issuing processor before it is advertised simultaneously to all other processors, \eg{}, in TSO and Alpha~\cite{alpha1998}.
    \item \emph{Non-atomic} (or non-multi-copy-atomic): a store becomes visible to different processors at different times, \eg{}, in POWER and ARM.
\end{itemize}
Multi-copy atomic stores are caused by the store buffer or write-through cache that is private to each processor.
Non-atomic stores arise (mostly) because of the sharing of a store buffer or a write-through cache by multiple processors, and such stores considerably complicate the formal definitions~\cite{sarkar2011understanding,flur2016modelling}.
WMM is an Alpha-like memory model which permits only multi-copy atomic stores and thus, prohibits shared store buffers or shared write-through caches in implementations.
\wmmSSB{} is an ARM/POWER-like memory model which admits non-atomic stores. 
We will present the implementations of both models using out-of-order (OOO) processors and cache-coherent memory systems. 
In particular, WMM and \wmmSSB{} allow the OOO processors in multicore settings to use all speculative techniques which are valid for uniprocessors, including even the \emph{load-value speculation}~\cite{lipasti1996value,Martin:2001:CIV:563998.564039,ghandour2010potential,perais2014eole,perais2014practical}, \emph{without} additional checks or logic.

We give operational definitions of both WMM and \wmmSSB{}.
An operational definition specifies an abstract machine, and the legal behaviors of a program under the memory model are those that can result by running the program on the abstract machine.
We observe a growing interest in operational definitions: memory models of x86, ARM and POWER have all been formalized operationally~\cite{owens2009better,sewell2010x86,sarkar2011understanding,sarkar2012synchronising,flur2016modelling}, and researchers are even seeking operational definitions for high-level languages like C++~\cite{Kang:2017:PSR:3009837.3009850}.
This is perhaps because all possible program results can be derived from operational definitions mechanically while axiomatic definitions require guessing the whole program execution at the beginning.
For complex programs with dependencies, loops and conditional branches, guessing the whole execution may become prohibitive.

Unfortunately, the operational models of ARM and POWER are too complicated because their abstract machines involve microarchitectural details like reorder buffers (ROBs), partial and speculative instruction execution, instruction replay on speculation failure, \etc{}
The operational definitions of WMM and \wmmSSB{} are much simpler because they are described in terms of \emph{Instantaneous Instruction Execution} (\IIE), which is the style used in the operational definitions of SC~\cite{lamport1979make} and TSO~\cite{owens2009better,sewell2010x86}.
An \IIE{} abstract machine consists of $n$ atomic processors and an $n$-ported \regfile{} memory.
The atomic processor executes instructions instantaneously and in order, so it always has the up-to-date architectural (register) state.
The \regfile{} memory executes loads and stores instantaneously.
Instruction reorderings and store atomicity/non-atomicity are captured by including different types of buffers between the processors and the \regfile{} memory, like the store buffer in the definition of TSO.
In the background, data moves between these buffers and the memory asynchronously, \eg, to drain a value from a store buffer to the memory.

\IIE{} definitions free programmers from reasoning partially executed instructions, which is unavoidable for ARM and POWER operational definitions.
One key tradeoff to achieve \IIE{} is to forbid a store to overtake a load, \ie, disallow Ld-St reordering.
Allowing such reordering requires each processor in the abstract machine to maintain multiple unexecuted instructions in order to see the effects of future stores, and the abstract machine has to contain the complicated ROB-like structures.
Ld-St reordering also complicates axiomatic definitions because it creates the possibility of ``out-of-thin-air'' behaviors~\cite{Boehm:2014:OGA:2618128.2618134}, which are impossible in any real implementation and consequently must be disallowed.
We also offer evidence, based on simulation experiments, that disallowing Ld-St reordering has no discernible impact on performance.

For a quick comparison, we summarize the properties of common memory models in Figure~\ref{tab: mm sumary}.
SC and TSO have simple definitions but forbid Ld-Ld and St-St reorderings, and consequently, are not candidates for \RISCV.
WMM is similar to RMO and Alpha but neither has an operational definition.
Also WMM has a simple axiomatic definition, while Alpha requires a complicated axiom to forbid out-of-thin-air behaviors (see Section~\ref{sec: wmm axiom}), and RMO has an incorrect axiom about data-dependency ordering (see Section~\ref{sec: alpha rmo rc bad}).

ARM, POWER, and \wmmSSB{} are similar models in the sense that they all admit non-atomic stores.
While the operational models of ARM and POWER are complicated, \wmmSSB{} has a simpler \IIE{} definition and allows competitive implementations (see Section~\ref{sec: ssb hw = model}).
The axiomatic models of ARM and POWER are also complicated: 
four relations in the POWER axiomatic model~\cite[Section 6]{alglave2014herding} are defined in a fixed point manner, \ie, their definitions mutually depend on each other.

Release Consistency (RC) are often mixed with the concept of ``SC for data-race-free (DRF) programs''~\cite{adve1990weak}.
It should be noted that ``SC for DRF'' is inadequate for an ISA memory model, which must specify behaviors of \emph{all} programs.
The original RC definition~\cite{gharachorloo1990memory} attempts to specify all program behaviors, and are more complex and subtle than the ``SC for DRF'' concept.
We show in Section~\ref{sec: alpha rmo rc bad} that the RC definition fails a litmus test for non-atomic stores and forbids shared write-through caches in implementation.

This paper makes the following contributions: 
\begin{enumerate}
    \item WMM, the \emph{first} weak memory model that is defined in \IIE{} and allows Ld-Ld reordering, and its axiomatic definition;
    \item \wmmSSB{}, an extension on WMM that admits non-atomic stores and has an \IIE{} definition;
    \item WMM and \wmmSSB{} implementations based on OOO processors that admit all uniprocessor speculative techniques (such as load-value prediction) without additional checks; 
    \item Introduction of \emph{invalidation buffers} in the \IIE{} definitional framework to model Ld-Ld and other reorderings.
\end{enumerate}

\noindent\textbf{Paper organization:}
Section~\ref{sec: related work} presents the related work.
Section~\ref{sec: litmus} gives litmus tests for distinguishing memory models.
Section~\ref{sec: I2E} introduces \IIE{}.
Section~\ref{sec: WMM} defines WMM.
Section~\ref{sec: wmm impl} shows the WMM implementation using OOO processors.
Section~\ref{sec: ld-st reorder} evaluates the performance of WMM and the influence of forbidding Ld-St reordering.
Section~\ref{sec: non atomic mem} defines \wmmSSB{}.
Section~\ref{sec: ssb impl} presents the \wmmSSB{} implementations with non-atomic stores.
Section~\ref{sec: alpha rmo rc bad} shows the problems of RC and RMO.
Section~\ref{sec: conclude} offers the conclusion.

\section{Related Work} \label{sec: related work}

SC~\cite{lamport1979make} is the simplest model, but naive implementations of SC suffer from poor performance.
Although researchers have proposed aggressive techniques to preserve SC~\cite{gharachorloo1991two,ranganathan1997using,guiady1999sc+,gniady2002speculative,ceze2007bulksc,wenisch2007mechanisms,blundell2009invisifence,singh2012end,lin2012efficient,gope2014atomic}, they are rarely adopted in commercial processors perhaps due to their hardware complexity.
Instead the manufactures and researchers have chosen to present weaker memory models, \eg{}, TSO~\cite{sparc1992sparcv8,owens2009better,sewell2010x86,Sarkar:2009:SXM:1594834.1480929}, PSO~\cite{weaver1994sparc}, RMO~\cite{weaver1994sparc}, Alpha~\cite{alpha1998}, Processor Consistency~\cite{goodman1991cache}, Weak Consistency~\cite{dubois1986memory}, RC~\cite{gharachorloo1990memory}, CRF~\cite{shen1999commit}, Instruction Reordering + Store Atomicity~\cite{arvind2006memory}, POWER~\cite{power2013version} and ARM~\cite{armv7ar}.
The tutorials by Adve \etal{} \cite{adve1996shared} and by Maranget \etal{} \cite{maranget2012tutorial} provide relationships among some of these models.


A large amount of research has also been devoted to specifying the memory models of high-level languages: C++ \cite{c++n4527,boehm2008foundations,batty2011mathematizing,Batty:2016:OSA:2914770.2837637,Kang:2017:PSR:3009837.3009850}, Java \cite{manson2005java,cenciarelli2007java, maessen2000improving}, \etc{}
We will provide compilation schemes from C++ to WMM and \wmmSSB{}.


Recently, Lustig \etal{} have used Memory Ordering Specification Tables (MOSTs) to describe memory models, and proposed a hardware scheme to dynamically convert programs across memory models described in MOSTs \cite{lustig2015armor}.
MOST specifies the ordering strength (\eg{}, locally ordered, multi-copy atomic) of two instructions from the same processor under different conditions (\eg{}, data dependency, control dependency).
Our work is orthogonal in that we propose new memory models with operational definitions.

\section{Memory Model Litmus Tests} \label{sec: litmus}

Here we offer two sets of litmus tests to highlight the differences between memory models regarding store atomicity and instruction reorderings, including enforcement of dependency-ordering.
All memory locations are initialized to 0.

\subsection{Store Atomicity Litmus Tests}

Figure~\ref{fig: store atomicity litmus} shows four litmus tests to distinguish between these three types of stores.
We have deliberately added data dependencies and Ld-Ld fences ($\LdLdFence$) to these litmus tests to prevent instruction reordering, \eg, the data dependency between $I_2$ and $I_3$ in Figure \ref{fig: dekker extend}.
Thus the resulting behaviors can arise only because of different store atomicity properties.
We use $\LdLdFence$ for memory models that can reorder data-dependent loads, \eg, $I_5$ in Figure~\ref{fig: wrc} would be the $\mathsf{MB}$ fence for Alpha.
For other memory models that order data-dependent loads (\eg, ARM), $\LdLdFence$ could be replaced by a data dependency (like the data dependency between $I_2$ and $I_3$ in Figure~\ref{fig: dekker extend}).
The Ld-Ld fences only stop Ld-Ld reordering; they do not affect store atomicity in these tests.

\begin{figure}[!htb]
    \centering
    \subfloat[SBE: test for multi-copy atomic stores\label{fig: dekker extend}]{
        \footnotesize
        \begin{tabular}{|p{23ex}|p{23ex}|}
            \hline
            {Proc. P1} & {Proc. P2} \\ 
            \hline
            $I_1: \StInst\ a\ 1$            & $I_4: \StInst\ b\ 1$ \\
            $I_2: r_1 = \LdInst\ a$         & $I_5: r_3 = \LdInst\ b$ \\
            $I_3: r_2 = \LdInst\ (b+r_1-1)$ & $I_6: r_4 = \LdInst\ (a+r_3-1)$ \\ \hline
            \multicolumn{2}{|l|}{SC forbids but TSO allows: $r_1 = 1, r_2 = 0, r_3 = 1, r_4 = 0$} \\ \hline
        \end{tabular}
    }\\
    \subfloat[WRC: test for non-atomic stores~\cite{sarkar2011understanding}\label{fig: wrc}]{
        \footnotesize
        \begin{tabular}{|l|l|l|}
            \hline
            Proc. P1 & Proc. P2 & Proc. P3 \\
            \hline
            $I_1: \StInst\ a\ 2$ & $I_2: r_1=\LdInst\ a$        & $I_4: r_2=\LdInst\ b$ \\
                                 & $I_3: \StInst\ b\ (r_1 - 1)$ & $I_5: \LdLdFence$ \\
                                 &                              & $I_6: r_3=\LdInst\ a$ \\
            \hline
            \multicolumn{3}{|p{42ex}|}{TSO, RMO and Alpha forbid, but RC, ARM and POWER allow: $r_1=2,\ r_2=1,\ r_3=0$} \\
            \hline
        \end{tabular}
    }\\
    \subfloat[WWC: test for non-atomic stores~\cite{maranget2012tutorial,wwctest}\label{fig: wwc}]{
        \footnotesize
        \begin{tabular}{|l|l|l|}
            \hline
            Proc. P1 & Proc. P2 & Proc. P3 \\
            \hline
            $I_1: \StInst\ a\ 2$ & $I_2: r_1=\LdInst\ a$        & $I_4: r_2=\LdInst\ b$ \\
                                 & $I_3: \StInst\ b\ (r_1 - 1)$ & $I_5: \StInst\ a\ r_2$ \\
            \hline
            \multicolumn{3}{|p{45ex}|}{TSO, RMO, Alpha and \textcolor{red}{\emph{RC}} forbid, but ARM and POWER allow: $r_1=2,\ r_2=1,\ m[a]=2$} \\
            \hline
        \end{tabular}
    }\\
    \subfloat[IRIW: test for non-atomic stores~\cite{sarkar2011understanding}\label{fig: iriw}]{
        \footnotesize
        \begin{tabular}{|l|l|l|l|}
            \hline
            Proc. P1 & Proc. P2 & Proc. P3 & Proc. P4 \\
            \hline
            $I_1: \StInst\ a\ 1$ & $I_2: r_1=\LdInst\ a$ & $I_5: \StInst\ b\ 1$  & $I_6: r_3=\LdInst\ b$ \\
                                 & $I_3: \LdLdFence$     &                       & $I_7: \LdLdFence$ \\
                                 & $I_4: r_2=\LdInst\ b$ &                       & $I_8: r_4=\LdInst\ a$ \\
            \hline
            \multicolumn{4}{|p{50ex}|}{TSO, RMO and Alpha forbid, but RC, ARM and POWER allow: $r_1=1,\ r_2=0,\ r_3=1,\ r_4=0$} \\
            \hline
        \end{tabular}
    }
    \caption{Litmus tests for store atomicity}\label{fig: store atomicity litmus}
\end{figure}

\noindent{\textbf{SBE:}}
In a machine with {single-copy atomic stores} (\eg{}, an SC machine), when both $I_2$ and $I_5$ have returned value 1, stores $I_1$ and $I_4$ must have been globally advertised.
Thus $r_2$ and $r_4$ cannot both be 0.
However, a machine with store buffers (\eg{}, a TSO machine) allows P1 to forward the value of $I_1$ to $I_2$ locally without advertising $I_1$ to other processors, violating the single-copy atomicity of stores.

\noindent{\textbf{WRC:}}
Assuming the store buffer is private to each processor (\ie{}, {multi-copy atomic stores}), if one observes $r_1=2$ and $r_2=1$ then $r_3$ must be 2.
However, if an architecture allows a store buffer to be shared by P1 and P2 but not P3, then P2 can see the value of $I_1$ from the shared store buffer before $I_1$ has updated the memory, allowing P3 to still see the old value of $a$.
A write-through cache shared by P1 and P2 but not P3 can cause this \emph{non-atomic store} behavior in a similar way, \eg, $I_1$ updates the shared write-through cache but has not invalidated the copy in the private cache of P3 before $I_6$ is executed.

\noindent{\textbf{WWC:}}
This litmus test is similar to WRC but replaces the load in $I_6$ with a store.
The behavior is possible if P1 and P2 share a write-through cache or store buffer.
However, RC forbids this behavior (see Section~\ref{sec: alpha rmo rc bad}).

\noindent{\textbf{IRIW:}}
This behavior is possible if P1 and P2 share a write-through cache or a store buffer and so do P3 and P4.

\subsection{Instruction Reordering Litmus Tests}\label{sec: reorder litmus}
Although processors fetch and commit instructions in order, speculative and out-of-order execution causes behaviors as if instructions were reordered.
Figure~\ref{fig: reorder litmus} shows the litmus tests on these reordering behaviors.

\begin{figure}[!htb]
    \centering
    \subfloat[SB: test for St-Ld reordering~\cite{maranget2012tutorial}\label{fig: dekker}]{
        \footnotesize
        \begin{tabular}{|l|l|}
            \hline
            {Proc. P1} & {Proc. P2} \\ 
            \hline
            $\!\! I_1: \StInst\ a\ 1 \!\!$    & $\!\! I_3: \StInst\ b\ 1 \!\!$ \\
            $\!\! I_2: r_1 = \LdInst\ b \!\!$ & $\!\! I_4: r_2 = \LdInst\ a \!\!$ \\ \hline
            \multicolumn{2}{|p{25ex}|}{SC forbids, but TSO allows: $r_1 = 0, r_2 = 0 \!\!$} \\ \hline
        \end{tabular}
    }
    \hspace{2pt}
    \subfloat[MP: test for Ld-Ld and St-St reorderings~\cite{sarkar2011understanding}\label{fig: msg pass}]{
        \footnotesize
        \begin{tabular}{|l|l|}
            \hline
            {Proc. P1} & {Proc. P2} \\ 
            \hline
            $\!\! I_1: \StInst\ a\ 1$  & $\!\! I_3: r_1 = \LdInst\ b \!\!$ \\
            $\!\! I_2: \StInst\ b\ 1$   & $\!\! I_4: r_2 = \LdInst\ a \!\!$ \\ 
            \hline
            \multicolumn{2}{|p{26ex}|}{TSO forbids, but Alpha and RMO allow: $r_1=1, r_2 = 0$} \\ 
            \hline
        \end{tabular}
    }\\
    \subfloat[LB: test for Ld-St reordering~\cite{sarkar2011understanding}\label{fig: ld buffer}]{
        \footnotesize
        \begin{tabular}{|l|l|}
            \hline
            {Proc. P1} & {Proc. P2} \\ 
            \hline
            $\!\! I_1: r_1 \!=\! \LdInst\ b \!\!$ & $\!\! I_3: r_2 \!=\! \LdInst\ a \!\!$ \\
            $\!\! I_2: \StInst\ a\ 1 \!\!$    & $\!\! I_4: \StInst\ b\ 1 \!\!$ \\ \hline
            \multicolumn{2}{|p{23.5ex}|}{TSO forbids, but Alpha, RMO, RC, POWER and ARM allow: $r_1 = r_2 = 1$} \\ \hline
        \end{tabular}
    }
    \hspace{2pt}
    \subfloat[MP+Ctrl: test for control-dependency ordering\label{fig: mp+ctrl}]{
        \footnotesize
        \begin{tabular}{|l|l|}
            \hline
            Proc. P1 & Proc. P2 \\
            \hline
            $\!\! I_1: \StInst\ a\ 1 \!\!$  & $\!\! I_4: r_1 = \LdInst\ b \!\!$ \\
            $\!\! I_2: \mathsf{FENCE} \!\!$ & $\!\! I_5: \mathsf{if}(r_1\!\neq\! 0)\ \mathsf{exit}\!\!\!$ \\
            $\!\! I_3: \StInst\ b\ 1 \!\!$  & $\!\! I_6: r_2 = \LdInst\ a \!\!$\\
            \hline
            \multicolumn{2}{|p{28.5ex}|}{Alpha, RMO, RC, ARM and POWER allow: $r_1=1,r_2=0$} \\
            \hline
        \end{tabular}
    }\\
    \subfloat[MP+Mem: test for memory-dependency ordering\label{fig: mp+mem}]{
        \footnotesize
        \begin{tabular}{|l|l|}
            \hline
            Proc. P1 & Proc. P2 \\
            \hline
            $\!\! I_1\!: \StInst\ a\ 1 \!\!$    & $\!\! I_4\!:\! r_1 = \LdInst\ b \!\!$ \\
            $\!\! I_2\!: \mathsf{FENCE} \!\!$   & $\!\! I_5\!:\! \StInst\ (r_1+a)\ 42\!\!\!$ \\
            $\!\! I_3\!: \StInst\ b\ 100 \!\!$  & $\!\! I_6\!:\! r_2 = \LdInst\ a \!\!$\\
            \hline
            \multicolumn{2}{|p{27ex}|}{Alpha, RMO, RC, ARM and POWER allow: $r_1=100, r_2=0$} \\
            \hline
        \end{tabular}
    }
    \hspace{2pt}
    \subfloat[MP+Data: test for data-dependency ordering \label{fig: mp+data}]{
        \footnotesize
        \begin{tabular}{|l|l|}
            \hline
            Proc. P1 & Proc. P2 \\
            \hline
            $\!\! I_1\!: \StInst\ a\ 1 \!\!$  & $\!\! I_4\!: r_1 = \LdInst\ b \!\!$ \\
            $\!\! I_2\!: \mathsf{FENCE} \!\!$ & $\!\! I_5\!: r_2 = \LdInst\ r_1 \!\!$ \\
            $\!\! I_3\!: \StInst\ b\ a \!\!$  & \\
            \hline
            \multicolumn{2}{|p{25ex}|}{RMO, RC, ARM and POWER forbid, but Alpha allows: $r_1=a,r_2=0$} \\
            \hline
        \end{tabular}
    }
    \caption{Litmus tests for instruction reorderings}\label{fig: reorder litmus}
\end{figure}

\noindent{\textbf{SB:}}
A TSO machine can execute $I_2$ and $I_4$ while $I_1$ and $I_3$ are buffered in the store buffers.
The resulting behavior is as if the store and the load were reordered on each processor.

\noindent{\textbf{MP:}}
In an Alpha machine, $I_1$ and $I_2$ may be drained from the store buffer of P1 out of order; $I_3$ and $I_4$ in the ROB of P2 may be executed out of order.
This is as if P1 reordered the two stores and P2 reordered the two loads.

\noindent{\textbf{LB:}}
Some machines may enter a store into the memory before all older instructions have been committed. 
This results in the Ld-St reordering shown in Figure~\ref{fig: ld buffer}.
Since instructions are committed in order and stores are usually not on the critical path, the benefit of the eager execution of stores is limited.
In fact we will show by simulation that Ld-St reordering does not improve performance (Section~\ref{sec: ld-st reorder}).

\noindent\textbf{MP+Ctrl:}
This test is a variant of MP.
The two stores in P1 must update memory in order due to the fence.
Although the execution of $I_6$ is conditional on the result of $I_4$, P2 can issue $I_6$ speculatively by predicting branch $I_5$ to be not taken.
The execution order $I_6, I_1, I_2, I_3, I_4, I_5$ results in $r_1=1$ and $r_2=0$.

\noindent\textbf{MP+Mem:}
This test replaces the control dependency in MP+Ctrl with a (potential) memory dependency, \ie, the unresolved store address of $I_5$ may be the same as the load address of $I_6$ before $I_4$ is executed, 
However, P2 can execute $I_6$ speculatively by predicting the addresses are not the same.
This results in having $I_6$ overtake $I_4$ and $I_5$.

\noindent{\textbf{MP+Data:}}
This test replaces the control dependency in MP+Ctrl with a data dependency, \ie, the load address of $I_5$ depends on the result of $I_4$.
A processor with \emph{load-value prediction}~\cite{lipasti1996value,Martin:2001:CIV:563998.564039,ghandour2010potential,perais2014eole,perais2014practical} may guess the result of $I_4$ before executing it, and issue $I_5$ speculatively.
If the guess fails to match the real execution result of $I_4$, then $I_5$ would be killed.
But, if the guess is right, then essentially the execution of the two data-dependent loads ($I_4$ and $I_5$) has been reordered.

\subsection{Miscellaneous Tests}
All programmers expect memory models to obey \emph{per-location SC}~\cite{cantin2003complexity}, \ie, all accesses to a single address appear to execute in a sequential order which is consistent with the program order of each thread (Figure~\ref{fig: corr}).

\begin{figure}[!htb]
    \centering
    \begin{minipage}[b]{0.45\columnwidth}
        \footnotesize
        \centering
        \begin{tabular}{|l|l|}
            \hline
            Proc. P1 & {Proc. P2$\!\!$} \\ 
            \hline
            $\!\! I_1: r_1 = \LdInst\ a \!\!$ & $\!\! I_3: \StInst\ a\ 1 \!\!$ \\ 
            $\!\! I_2: r_2 = \LdInst\ a \!\!$ & \\ 
            \hline
            \multicolumn{2}{|p{24ex}|}{Models with per-location SC forbid: $r_1 = 1,r_2=0$} \\
            \hline
        \end{tabular}
        \caption{Per-location SC} \label{fig: corr}
    \end{minipage}
    \hspace{1pt}
    \begin{minipage}[b]{0.52\columnwidth}
        \centering
        \footnotesize
        \begin{tabular}{|l|l|}
            \hline
            {Proc. P1} & {Proc. P2} \\ 
            \hline
            $\!\! I_1: r_1 = \LdInst\ b \!\!$ & $\!\! I_3: r_2 = \LdInst\ a \!\!$ \\
            $\!\! I_2: \StInst\ a\ r_1 \!\!$    & $\!\! I_4: \StInst\ b\ r_2 \!\!$ \\ \hline
            \multicolumn{2}{|l|}{\hspace{-3pt}All models forbid: $r_1 = r_2 = 42 \!\!$} \\ \hline
        \end{tabular}
        \caption{Out-of-thin-air read} \label{fig: no thin air}
    \end{minipage}
\end{figure}

Out-of-thin-air behaviors (Figure~\ref{fig: no thin air}) are impossible in real implementations. 
Sometimes such behaviors are permitted by axiomatic models due to incomplete axiomatization.

\section{Defining Memory Models in I\textsuperscript{2}E} 
\label{sec: I2E}

Figure~\ref{fig: 3 op models} shows the \IIE{} abstract machines for SC, TSO/PSO and WMM models.
All abstract machines consist of $n$ atomic processors and an $n$-ported \regfile{} memory $m$.
Each processor contains a register state $s$, which represents all architectural registers, including both the general purpose registers and special purpose registers, such as PC.
The abstract machines for TSO/PSO and WMM also contain a \emph{store buffer} $sb$ for each processor, and the one for WMM also contains an \emph{invalidation buffer} $ib$ for each processor as shown in the figure.
In the abstract machines all buffers are unbounded.
The operations of these buffers will be explained shortly.


\begin{figure}[!htb]
	\centering
	\subfloat[SC\label{fig: sc op model}]{\includegraphics[width=0.3\columnwidth]{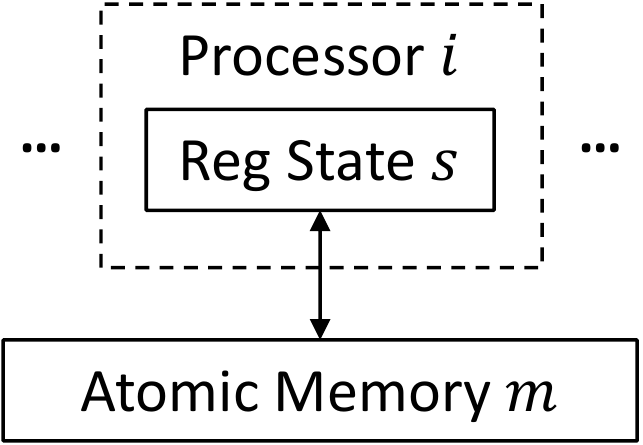}}\hspace{1pt}
    \subfloat[TSO/PSO\label{fig: tso op model}]{\includegraphics[width=0.3\columnwidth]{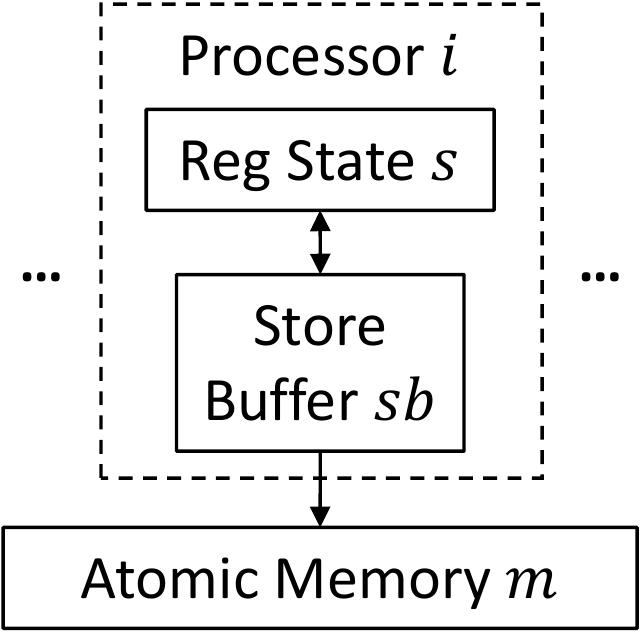}}\hspace{1pt}
    \subfloat[WMM\label{fig: wmm op model}]{\includegraphics[width=0.35\columnwidth]{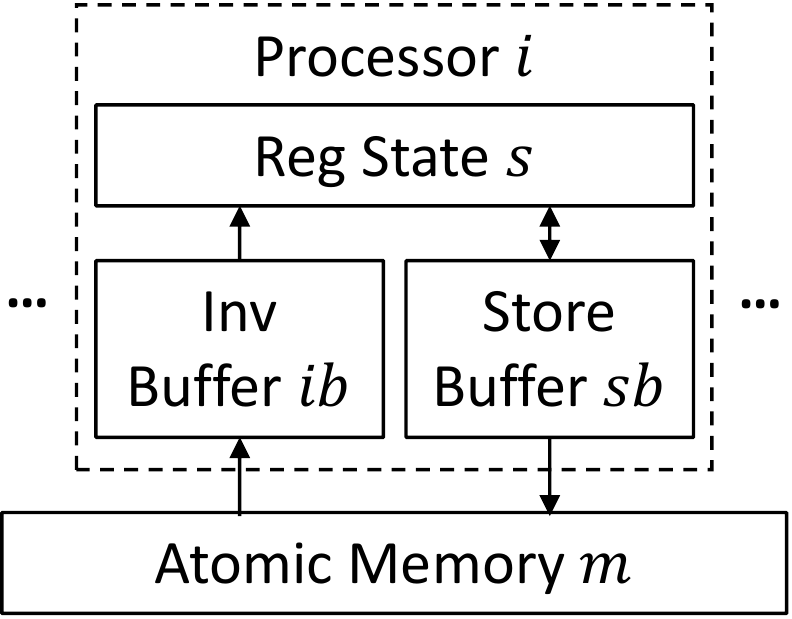}}
    \caption{\IIE{} abstract machines for different models}\label{fig: 3 op models}
\end{figure}

The operations of the SC abstract machine are the simplest: in one step we can select any processor to execute the next instruction on that processor \emph{atomically}.
That is, if the instruction is a non-memory instruction (\eg, ALU or branch), it just modifies the register states of the processor; if it is a load, it reads from the \regfile{} memory instantaneously and updates the register state; and if it is a store, it updates the \regfile{} memory instantaneously and increments the PC.

\subsection{TSO Model}

The TSO abstract machine proposed in~\cite{owens2009better,sewell2010x86} (Figure \ref{fig: tso op model}) contains a store buffer $sb$ for each processor. 
Just like SC, any processor can execute an instruction atomically, and if the instruction is a non-memory instruction, it just modifies the local register state.
A store is executed by inserting its $\langle \mathrm{address, value}\rangle$ pair into the local $sb$ instead of writing the data in memory.
A load first looks for the load address in the local $sb$ and returns the value of the \emph{youngest} store for that address.
If the address is not in the local $sb$, then the load returns the value from the \regfile{} memory.
TSO can also perform a \emph{background} operation, which removes the \emph{oldest} store from a $sb$ and writes it into the \regfile{} memory.
As we discussed in Section~\ref{sec: litmus}, store buffer allows TSO to do St-Ld reordering, i.e., pass the SB litmus test (Figure~\ref{fig: dekker}). 

In order to enforce ordering in accessing the memory and to rule out non-SC behaviors, TSO has a fence instruction, which we refer to as $\ComInst$.
When a processor executes a $\ComInst$ fence, it gets blocked unless its $sb$ is empty.
Eventually, any $sb$ will become empty as a consequence of the background operations that move data from the $sb$ to the memory.
For example, we need to insert a $\ComInst$ fence after each store in Figure~\ref{fig: dekker} to forbid the non-SC behavior in TSO.

We summarize the operations of the TSO abstract machine in Figure~\ref{fig: tso operations}.
Each operation consists of a \emph{predicate} and an \emph{action}.
The operation can be performed by taking the action only when the predicate is true.
Each time we perform only one operation (either instruction execution or $sb$ dequeue) atomically in the whole system  (\eg{}, no two processors can execute instructions simultaneously).
The choice of which operation to perform is nondeterministic.

\begin{figure}[!htb]
    \begin{boxedminipage}{\columnwidth}
        \small
        \textbf{\tsoNmRule} (non-memory execution)\\
        \emph{Predicate:} The next instruction of a processor is a non-memory instruction. \\
        \emph{Action:} Instruction is executed by local computation.
        \\
        \textbf{\tsoLdRule} (load execution) \\
        \emph{Predicate:} The next instruction of a processor is a load. \\
        \emph{Action:} Assume the load address is $a$.
        The load returns the value of the \emph{youngest} store for $a$ in $sb$ if $a$ is present in the $sb$ of the processor, otherwise, the load returns $m[a]$, \ie{}, the value of address $a$ in the \regfile{} memory.
        \\
        \textbf{\tsoStRule} (store execution) \\
        \emph{Predicate:} The next instruction of a processor is a store. \\
        \emph{Action:} Assume the store address is $a$ and the store value is $v$.
        The processor inserts the store $\langle a,v\rangle$ into its $sb$.
        \\
        \textbf{\tsoComRule} ($\ComInst$ execution) \\
        \emph{Predicate:} The next instruction of a processor is a $\ComInst$ and the $sb$ of the processor is empty. \\
        \emph{Action:} The $\ComInst$ fence is executed simply as a NOP.
        \\
        \textbf{\tsoDeqSbRule} (background store buffer dequeue) \\
        \emph{Predicate:} The $sb$ of a processor is not empty.\\
        \emph{Action:} Assume the $\langle \mathrm{address, value}\rangle$ pair of the {oldest} store in the $sb$ is $\langle a,v\rangle$.
        Then this store is removed from $sb$, and the \regfile{} memory $m[a]$ is updated to $v$.
    \end{boxedminipage}\vspace{1pt}
    \begin{boxedminipage}{\columnwidth}
        \small
        \textbf{\psoDeqSbRule} (background store buffer dequeue) \\
        \emph{Predicate:} The $sb$ of a processor is not empty.\\
        \emph{Action:} Assume the value of the {oldest} store for \emph{some} address $a$ in the $sb$ is $v$.
        Then this store is removed from $sb$, and the \regfile{} memory $m[a]$ is updated to $v$.
    \end{boxedminipage}
    \caption{Operations of the TSO/PSO abstract machine}\label{fig: tso operations}
\end{figure}

\noindent\textbf{Enabling St-St reordering:}
We can extend TSO to PSO by changing the background operation to dequeue the oldest store for \emph{any} address in $sb$ (see the \psoDeqSbRule{} operation in Figure~\ref{fig: tso operations}).
This extends TSO by permitting St-St reordering.

\section{WMM Model} \label{sec: WMM}

WMM allows Ld-Ld reordering in addition to the reorderings allowed by PSO.
Since a reordered load may read a stale value, we introduce a conceptual device called \emph{invalidation buffer}, $ib$, for each processor in the \IIE{} abstract machine (see Figure \ref{fig: wmm op model}).
$ib$ is an unbounded buffer of $\langle \mathrm{address, value} \rangle$ pairs, each representing a stale memory value for an address that can be observed by the processor.
Multiple stale values for an address in $ib$ are kept ordered by their staleness. 

The operations of the WMM abstract machine are similar to those of PSO except for the background operation and the load execution.
When the background operation moves a store from $sb$ to the \regfile{} memory, the original value in the \regfile{} memory, \ie{}, the stale value, enters the $ib$ of every other processor.
A load first searches the local $sb$.
If the address is not found in $sb$, it either reads the value in the \regfile{} memory or any stale value for the address in the local $ib$, the
choice between the two being nondeterministic.

The abstract machine operations maintain the following invariants:
once a processor \emph{observes} a store, it cannot observe any staler store for that address. 
Therefore, (1) when a store is executed, values for the store address in the local $ib$ are purged; (2) when a load is executed, values staler than the load result are flushed from the local $ib$; and (3) the background operation does not insert the stale value into the $ib$ of a processor if the $sb$ of the processor contains the address.

Just like introducing the $\ComInst$ fence in TSO, to prevent loads from reading the stale values in $ib$, we introduce the $\RecInst$ fence to clear the local $ib$. 
Figure~\ref{fig: wmm operations} summarizes the operations of the WMM abstract machine.

\begin{figure}[!htb]
    \begin{boxedminipage}{\columnwidth}
        \small
        \textbf{\wmmNmRule} (non-memory execution): Same as \tsoNmRule.
        \\
        \textbf{\wmmLdRule} (load execution) \\
        \emph{Predicate:} The next instruction of a processor is a load. \\
        \emph{Action:} Assume the load address is $a$.
        If $a$ is present in the $sb$ of the processor, then the load returns the value of the \emph{youngest} store for $a$ in the local $sb$.
        Otherwise, the load is executed in either of the following two ways (the choice is arbitrary):
        \begin{enumerate}
            \item The load returns the \regfile{} memory value $m[a]$, and all values for $a$ in the local $ib$ are removed.
            \item The load returns some value for $a$ in the local $ib$, and all values for $a$ {older} than the load result are removed from the local $ib$.
            (If there are multiple values for $a$ in $ib$, the choice of which one to read is arbitrary).
        \end{enumerate}
        \textbf{\wmmStRule} (store execution) \\
        \emph{Predicate:} The next instruction of a processor is a store. \\
        \emph{Action:} Assume the store address is $a$ and the store value is $v$.
        The processor inserts the store $\langle a,v\rangle$ into its $sb$, and removes all values for $a$ from its $ib$.
        \\
        \textbf{\wmmComRule} ($\ComInst$ execution): Same as \tsoComRule.
        \\
        \textbf{\wmmRecRule} (execution of a $\RecInst$ fence) \\
        \emph{Predicate:} The next instruction of a processor is a $\RecInst$. \\
        \emph{Action:} All values in the $ib$ of the processor are removed.
        \\
        \textbf{\wmmDeqSbRule} (background store buffer dequeue) \\
        \emph{Predicate:} The $sb$ of a processor is not empty.\\
        \emph{Action:} Assume the value of the {oldest} store for \emph{some} address $a$ in the $sb$ is $v$.
        First, the stale $\langle \mathrm{address, value}\rangle$ pair $\langle a, m[a]\rangle$ is inserted to the $ib$ of every other processor whose $sb$ does not contain $a$.
        Then this store is removed from $sb$, and $m[a]$ is set to $v$.
    \end{boxedminipage}
    \caption{Operations of the WMM abstract machine}\label{fig: wmm operations}
\end{figure}

\subsection{Properties of WMM}\label{sec: wmm litmus test}

Similar to TSO/PSO, WMM allows St-Ld and St-St reorderings because of $sb$ (Figures~\ref{fig: dekker} and \ref{fig: msg pass}).
To forbid the behavior in Figure~\ref{fig: dekker}, we need to insert a $\ComInst$ followed by a $\RecInst$ after the store in each processor.
$\RecInst$ is needed to prevent loads from getting stale values from $ib$.
The \IIE{} definition of WMM automatically forbids Ld-St reordering (Figure~\ref{fig: ld buffer}) and out-of-thin-air behaviors (Figure~\ref{fig: no thin air}).

\noindent\textbf{Ld-Ld reordering:}
WMM allows the behavior in Figure~\ref{fig: msg pass} due to St-St reordering.
Even if we insert a $\ComInst$ between the two stores in P1, the behavior is still allowed because $I_4$ can read the stale value 0 from $ib$.
This is as if the two loads in P2 were reordered.
Thus, we also need a $\RecInst$ between the two loads in P2 to forbid this behavior in WMM.

\noindent\textbf{No dependency ordering:}
WMM does not enforce any dependency ordering.
For example, WMM allows the behaviors of litmus tests in Figures~\ref{fig: mp+ctrl}, \ref{fig: mp+mem} and \ref{fig: mp+data} ($I_2$ should be $\ComInst$ in case of WMM), because the last load in P2 can always get the stale value 0 from $ib$ in each test.
Thus, it requires $\RecInst$ fences to enforce dependency ordering in WMM.
In particular, WMM can reorder the data-dependent loads (\ie, $I_4$ and $I_5$) in Figure~\ref{fig: mp+data}.

\noindent\textbf{Multi-copy atomic stores:}
Stores in WMM are multi-copy atomic, and WMM allows the behavior in Figure~\ref{fig: dekker extend} even when $\RecInst$ fences are inserted between Ld-Ld pairs $\langle I_2,I_3\rangle$ and $\langle I_5,I_6\rangle$.
This is because a store can be read by a load from the same processor while the store is in $sb$.
However, if the store is ever pushed from $sb$ to the \regfile{} memory, it becomes visible to all other processors simultaneously.
Thus, WMM forbids the behaviors in Figures~~\ref{fig: wrc}, \ref{fig: wwc} and \ref{fig: iriw} ($\LdLdFence$ should be $\RecInst$ in these tests).

\noindent\textbf{Per-location SC:}
WMM enforces per-location SC (Figure~\ref{fig: corr}), because both $sb$ and $ib$ enforce FIFO on same address entries.

\subsection{Axiomatic Definition of WMM}\label{sec: wmm axiom}

Based on the above properties of WMM, we give a simple axiomatic definition for WMM in Figure~\ref{fig: wmm axiomatic model} in the style of the axiomatic definitions of TSO and Alpha.
A $\True$ entry in the order-preserving table (Figure~\ref{fig: wmm order function}) indicates that if instruction $X$ precedes instruction $Y$ in the program order ($X\ProgOrd Y$) then the order must be maintained in the global memory order ($\MemOrd$).
$\MemOrd$ is a total order of all the memory and fence instructions from all processors. 
The notation $S\ReadFrom L$ means a load $L$ reads from a store $S$.
The notation $\max_{<mo} \{ \mathrm{set\ of\ stores} \}$ means the youngest store in the set according to $\MemOrd$.
The axioms are self-explanatory: the program order must be maintained if the order-preserving table says so, and a load must read from the youngest store among all stores that precede the load in either the memory order or the program order.
(See Appendix~\ref{sec:proof} for the equivalence proof of the axiomatic and \IIE{} definitions.)

These axioms also hold for Alpha with a slightly different order-preserving table, which marks the (Ld,St) entry as $a=b$.
(Alpha also merges $\ComInst$ and $\RecInst$ into a single fence).
However, allowing Ld-St reordering creates the possibility of out-of-thin-air behaviors, and Alpha uses an additional complicated axiom to disallow such behaviors~\cite[Chapter 5.6.1.7]{alpha1998}.
This axiom requires considering all possible execution paths to determine if a store is ordered after a load by dependency, while normal axiomatic models only examine a single execution path at a time.
Allowing Ld-St reordering also makes it difficult to define Alpha operationally.

\begin{figure}[!htb]
    \centering
    \subfloat[Axioms for WMM\label{fig: wmm axioms}]{
        \begin{boxedminipage}{\columnwidth}
            \small
            \textbf{Axiom \AxiInstOrd} (preserved instruction ordering):\\
            \centerline{$X \ProgOrd Y\ \wedge\ \orderFunc(X,Y) \Rightarrow\ X \MemOrd Y$}\\
            \textbf{Axiom \AxiLdVal} (the value of a load):\\
            \centerline{
                \begin{tabular}{l}
                $\StInst\ a\ v \ReadFrom \LdInst\ a\ \Rightarrow \hspace{2pt}\StInst\ a\ v =$ \\
                $\max_{<mo}\{\StInst\ a\ v' |\ \StInst\ a\ v'\MemOrd \LdInst\ a\ \vee\ \StInst\ a\ v'\ProgOrd\ \LdInst\ a \}$ \\
                \end{tabular}
            }
        \end{boxedminipage}
    }\\
    \subfloat[WMM order-preserving table, \ie{} $\orderFunc(X,Y)$ where $X\ProgOrd Y$
    \label{fig: wmm order function}]{
        \footnotesize
        \begin{tabular}{|p{10ex}|p{10ex}|p{10ex}|p{10ex}|p{10ex}|}
            \hline
            \backslashbox{X}{Y} & $\LdInst\ b$ & $\StInst\ b\ v'$ & $\RecInst$ & $\ComInst$ \\ \hline
            $\LdInst\ a$        & $a = b$      & $\True$          & $\True$    & $\True$     \\ \hline
            $\StInst\ a\ v$     & $\False$     & $a = b$          & $\False$   & $\True$     \\ \hline
            $\RecInst$          & $\True$      & $\True$          & $\True$    & $\True$     \\ \hline
            $\ComInst$          & $\False$     & $\True$          & $\True$    & $\True$     \\ \hline
        \end{tabular}
    }
    \caption{Axiomatic definition of WMM}\label{fig: wmm axiomatic model}
\end{figure}


\subsection{Compiling C++11 to WMM} \label{sec: wmm c++}

C++ primitives \cite{c++n4527} can be mapped to WMM instructions in an efficient way as shown in Figure~\ref{tab: c++ map}.
For the purpose of comparison, we also include a mapping to POWER~\cite{batty2012clarifying}. 

\begin{figure}[!htb]
	\centering
	\footnotesize
	\begin{tabular}{|l|l|l|}
		\hline
		C++ operations & WMM instructions & POWER instructions \\
		\hline
		Non-atomic Load  & $\LdInst$             & $\LdInst$ \\ \hline
		Load Relaxed     & $\LdInst$             & $\LdInst$ \\ \hline
		Load Consume     & $\LdInst;\ \RecInst$  & $\LdInst$ \\ \hline
		Load Acquire     & $\LdInst;\ \RecInst$  & $\LdInst;\ \cmpInst;\ \bcInst;\ \isyncInst$ \\ \hline
		Load SC          & $\ComInst;\ \RecInst;$ & $\syncInst;\ \LdInst;\ \cmpInst;$ \\ 
		                 & $\LdInst;\ \RecInst$  & $\bcInst;\ \isyncInst$ \\ \hline
		Non-atomic Store & $\StInst$             & $\StInst$ \\ \hline
		Store Relaxed    & $\StInst$             & $\StInst$ \\ \hline
		Store Release    & $\ComInst;\ \StInst$  & $\lwsyncInst;\ \StInst \!\!$ \\ \hline
		Store SC         & $\ComInst;\ \StInst$  & $\syncInst;\ \StInst$ \\
		\hline
	\end{tabular}
	\caption{Mapping C++ to WMM and POWER} \label{tab: c++ map}
\end{figure}

The $\ComInst;\RecInst$ sequence in WMM is the same as a  $\syncInst$ fence in POWER, and $\ComInst$ is similar to $\lwsyncInst$.
The $\cmpInst;\bcInst;\isyncInst$ sequence in POWER serves as a Ld-Ld fence, so it is similar to a $\RecInst$ fence in WMM.
In case of Store SC in C++, WMM uses a $\ComInst$ while POWER uses a $\syncInst$, so WMM effectively saves one $\RecInst$.
On the other hand, POWER does not need any fence for Load Consume in C++, while WMM requires a $\RecInst$.


Besides the C++ primitives, a common programming paradigm is the well-synchronized program, in which all critical sections are protected by locks.
To maintain SC behaviors for such programs in WMM, we can add a $\RecInst$ after acquiring the lock and a $\ComInst$ before releasing the lock.

For any program, if we insert a $\ComInst$ before every store and insert a $\ComInst$ followed by a $\RecInst$ before every load, then the program behavior in WMM is guaranteed to be sequentially consistent.
This provides a conservative way for inserting fences when performance is not an issue.

\section{WMM Implementation} \label{sec: wmm impl}

WMM can be implemented using conventional OOO multiprocessors, and even the most aggressive speculative techniques cannot step beyond WMM.
To demonstrate this, we describe an OOO implementation of WMM, and show simultaneously how the WMM model (\ie, the \IIE{} abstract machine) captures the behaviors of the implementation.
The implementation is described abstractly to skip unrelated details (\eg{}, ROB entry reuse). 
The implementation consists of $n$ OOO processors and a coherent write-back cache hierarchy which we discuss next.


\begin{figure}[!htb]
    \centering
    \includegraphics[width=0.8\columnwidth]{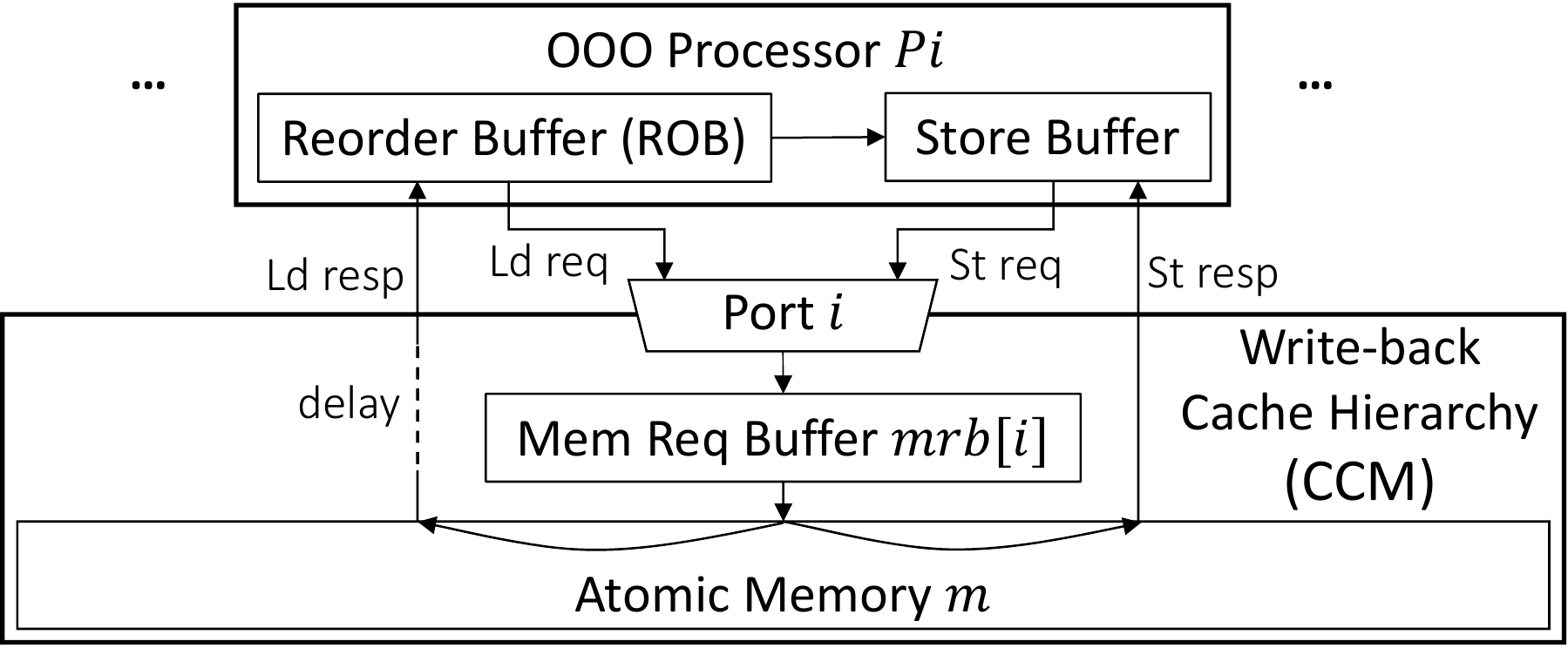}
    \caption{\ccmModelName+OOO: implementation of WMM} \label{fig: pow+ccm}
\end{figure}

\subsection{Write-Back Cache Hierarchy (\ccmModelName)}
We describe \ccmModelName{} as an abstraction of a conventional write-back cache hierarchy to avoid too much details.
In the following, we explain the function of such a cache hierarchy, abstract it to \ccmModelName, and relate \ccmModelName{} to the WMM model.

Consider a real $n$-ported write-back cache hierarchy with each port $i$ connected to processor $Pi$.
A request issued to port $i$ may be from a load instruction in the ROB of $Pi$ or a store in the store buffer of $Pi$. 
In conventional coherence protocols, all memory requests can be \emph{serialized}, \ie{}, each request can be considered as taking effect at some time point within its processing period~\cite{Vijayaraghavan2015}.
For example, consider the non-stalling MSI directory protocol in the Primer by Sorin \etal{}~\cite[Chapter 8.7.2]{sorin2011primer}.
In this protocol, a load request takes effect immediately if it hits in the cache; otherwise, it takes effect when it gets the data at the directory or a remote cache with M state.
A store request always takes effect at the time of writing the cache, \ie, either when it hits in the cache, or when it has received the directory response and all invalidation responses in case of miss.
We also remove the requesting store from the store buffer when a store request takes effect.
Since a cache cannot process multiple requests to the same address simultaneously, we assume requests to the same address from the same processor are processed in the order that the requests are issued to the cache.

\ccmModelName{} (Figure \ref{fig: pow+ccm}) abstracts the above cache hierarchy by operating as follows: every new request from port $i$ is inserted into a memory request buffer $mrb[i]$, which keeps requests to the same address in order; at any time we can remove the oldest request for an address from a $mrb$, let the request access the \regfile{} memory $m$, and either send the load result to ROB (which may experience a delay) or immediately dequeue the store buffer.
$m$ represents the coherent memory states.
Removing a request from $mrb$ and accessing $m$ captures the moment when the request takes effect.

It is easy to see that the \regfile{} memory in \ccmModelName{} corresponds to the \regfile{} memory in the WMM model, because they both hold the coherent memory values.
We will show shortly that how WMM captures the combination of \ccmModelName{} and OOO processors.
Thus any coherent protocol that can be abstracted as \ccmModelName{} can be used to implement WMM.

\subsection{Out-of-Order Processor (OOO)}
The major components of an OOO processor 
are the ROB and the store buffer (see Figure~\ref{fig: pow+ccm}).
Instructions are fetched into and committed from ROB in order; loads can be issued (\ie{}, search for data forwarding and possibly request \ccmModelName{}) as soon as its address is known; a store is enqueued into the store buffer only when the store commits (\ie, entries in a store buffer cannot be killed).
To maintain the per-location SC property of WMM, when a load $L$ is issued, it kills younger loads which have been issued but do not read from stores younger than $L$.
Next we give the correspondence between OOO and WMM.

\noindent\textbf{Store buffer:}
The state of the store buffer in OOO is represented by the $sb$ in WMM.
Entry into the store buffer when a store commits in OOO corresponds to the \wmmStRule{} operation.
In OOO, the store buffer only issues the oldest store for some address to \ccmModelName.
The store is removed from the store buffer when the store updates the \regfile{} memory in \ccmModelName.
This corresponds to the \wmmDeqSbRule{} operation.

\noindent\textbf{ROB and eager loads:}
Committing an instruction from ROB corresponds to executing it in WMM, and thus the architectural register state in both WMM and OOO must match at the time of commit.
Early execution of a load $L$ to address $a$ with a return value $v$ in OOO can be understood by considering where $\langle a,v \rangle$ resides in OOO when $L$ \emph{commits}.
Reading from $sb$ or \regfile{} memory $m$ in the \wmmLdRule{} operation covers the cases that $\langle a,v \rangle$ is, respectively, in the store buffer or the \regfile{} memory of \ccmModelName{} when $L$ commits.
Otherwise $\langle a,v \rangle$ is no longer present in \ccmModelName+OOO at the time of load commit and must have been overwritten in the \regfile{} memory of \ccmModelName.
This case corresponds to having performed the \wmmDeqSbRule{} operation to insert $\langle a,v \rangle$ into $ib$ previously, and now using the \wmmLdRule{} operation to read $v$ from $ib$.

\noindent\textbf{Speculations:}
OOO can issue a load speculatively by aggressive predictions, such as branch prediction (Figure~\ref{fig: mp+ctrl}), memory dependency prediction (Figure~\ref{fig: mp+mem}) and even load-value prediction (Figure~\ref{fig: mp+data}).
As long as all predictions related to the load eventually turn out to be correct, the load result got from the speculative execution can be preserved.
\emph{No further check is needed}.
Speculations effectively reorder dependent instructions, \eg{}, load-value speculation reorders data-dependent loads.
Since WMM does not require preserving any dependency ordering, speculations will neither break WMM nor affect the above correspondence between OOO and WMM.

\noindent\textbf{Fences:}
Fences never go into store buffers or \ccmModelName{} in the implementation.
In OOO, a $\ComInst$ can commit from ROB only when the local store buffer is empty.
$\RecInst$ plays a different role; at the time of commit it is a NOP, but while it is in the ROB, it stalls all younger loads (unless the load can bypass directly from a store which is younger than the $\RecInst$).
The stall prevents younger loads from reading values that would become stale when the $\RecInst$ commits.
This corresponds to clearing $ib$ in WMM.

\noindent\textbf{Summary:}
For any execution in the \ccmModelName+OOO implementation, we can operate the WMM model following the above correspondence.
Each time \ccmModelName+OOO commits an instruction $I$ from ROB or dequeues a store $S$ from a store buffer to memory, the \regfile{} memory of \ccmModelName, store buffers, and the results of committed instructions in \ccmModelName+OOO are exactly the same as those in the WMM model when the WMM model executes $I$ or dequeues $S$ from $sb$, respectively.

\section{Performance Evaluation of WMM} \label{sec: ld-st reorder}

We evaluate the performance of implementations of WMM, Alpha, SC and TSO. 
All implementations use OOO cores and coherent write-back cache hierarchy.
Since Alpha allows Ld-St reordering, the comparison of WMM and Alpha will show whether such reordering affects performance.

\subsection{Evaluation Methodology}
We ran SPLASH-2x benchmarks~\cite{woo1995splash,splash2x} on an 8-core multiprocessor using the ESESC simulator~\cite{ardestani2013esesc}.
We ran all benchmarks except ocean\_ncp, which allocates too much memory and breaks the original simulator.
We used sim-medium inputs except for cholesky, fft and radix, where we used sim-large inputs.
We ran all benchmarks to completion without sampling.

The configuration of the 8-core multiprocessor is shown in Figures~\ref{tab: sys config} and \ref{tab: core config} .
We do not use load-value speculation in this evaluation.
The Alpha implementation can mark a younger store as committed when instruction commit is stalled, as long as the store can never be squashed and the early commit will not affect single-thread correctness.
A committed store can be issued to memory or merged with another committed store in WMM and Alpha.
SC and TSO issue loads speculatively and monitor L1 cache evictions to kill speculative loads that violate the consistency model.
We also implement store prefetch as an optional feature for SC and TSO;
We use \textbf{SC-pf} and \textbf{TSO-pf} to denote the respective implementations with store prefetch.

\begin{figure}[!htb]
    \centering
    \begin{minipage}{\columnwidth}
        \centering
        \footnotesize
        \begin{tabular}{|l|l|}
            \hline
            Cores & 8 cores (@2GHz) with private L1 and L2 caches  \\
            \hline
            L3 cache & 4MB shared, MESI coherence, 64-byte cache line \\
            & 8 banks, 16-way, LRU replacement, max 32 req per bank  \\
            & 3-cycle tag, 10-cycle data (both pipelined) \\
            & 5 cycles between cache bank and core (pipelined) \\
            \hline
            Memory & 120-cycle latency, max 24 requests \\
            \hline
        \end{tabular}
        \caption{Multiprocessor system configuration} \label{tab: sys config}
    \end{minipage}
    \vspace{5pt}\\
    \begin{minipage}{\columnwidth}
        \centering
        \footnotesize
        \begin{tabular}{|l|l|}
            \hline
            Frontend & fetch + decode + rename,  7-cycle pipelined latency in all \\
            & 2-way superscalar, hybrid branch predictor \\
            \hline
            ROB & 128 entries, 2-way issue/commit \\ 
            \hline
            Function  & 2 ALUs, 1 FPU, 1 branch unit, 1 load unit, 1 store unit \\
            units &   32-entry reservation station per unit \\
            \hline
            Ld queue & Max 32 loads   \\ 
            \hline
            St queue  & Max 24 stores, containing speculative and committed stores\hspace{-3pt} \\
            \hline
            L1 D  & 32KB private, 1 bank, 4-way, 64-byte cache line \\ 
            cache & LRU replacement, 1-cycle tag, 2-cycle data (pipelined) \\
            & Max 32 upgrade and 8 downgrade requests \\
            \hline
            L2 cache & 128KB private, 1 bank, 8-way, 64-byte cache line \\
            & LRU replacement, 2-cycle tag, 6-cycle data (both pipelined) \\
            & Max 32 upgrade and 8 downgrade requests \\
            \hline
        \end{tabular}
        \caption{Core configuration} \label{tab: core config}
    \end{minipage}
\end{figure}

\subsection{Simulation Results}

\begin{figure*}[!t]
    \centering
    \includegraphics[width=\textwidth]{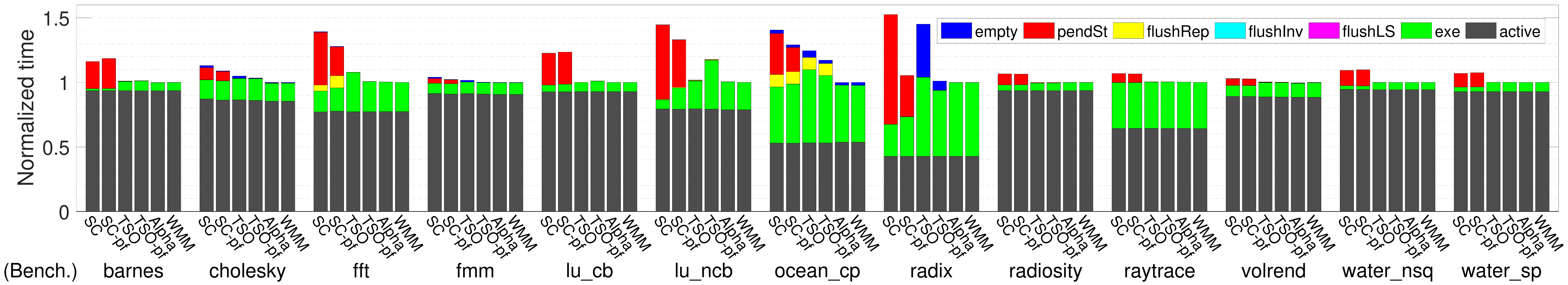}
    \caption{Normalized execution time and its breakdown at the commit slot of ROB}\label{fig: splash retire BW}
\end{figure*}

A common way to study the performance of memory models is to monitor the commit of instructions at the commit slot of ROB (\ie, the oldest ROB entry).
Here are some reasons why an instruction may not commit in a given cycle:
\begin{itemize}
    \item \textbf{empty}: The ROB is empty.
    \item \textbf{exe}: The instruction at the commit slot is still executing.
    \item \textbf{pendSt}: The load (in SC) or $\ComInst$ (in TSO, Alpha and WMM) cannot commit due to pending older stores.
    \item \textbf{flushLS}: ROB is being flushed because a load is killed by another older load (only in WMM and Alpha) or older store (in all models) to the same address.
    \item \textbf{flushInv}: ROB is being flushed after cache invalidation caused by a remote store (only in SC or TSO).
    \item \textbf{flushRep}: ROB is being flushed after cache replacement (only in SC or TSO).
\end{itemize}
Figure \ref{fig: splash retire BW} shows the execution time (normalized to WMM) and its breakdown at the commit slot of ROB. 
The total height of each bar represents the normalized execution time, and stacks represent different types of stall times added to the active committing time at the commit slot.

\noindent\textbf{WMM versus SC:}
WMM is much faster than both SC and SC-pf for most benchmarks, because a pending older store in the store queue can block SC from committing loads.

\noindent\textbf{WMM versus TSO:}
WMM never does worse than TSO or TSO-pf, and in some cases it shows up to 1.45$\times$ speedup over TSO (in radix) and 1.18$\times$ over TSO-pf (in lu\_ncb).
There are two disadvantages of TSO compared to WMM.
First, load speculation in TSO is subject to L1 cache eviction, \eg, in benchmark ocean\_cp.
Second, TSO requires prefetch to reduce store miss latency, \eg, a full store queue in TSO stalls issue to ROB and makes ROB empty in benchmark radix.
However, prefetch may sometimes degrade performance due to interference with load execution, \eg, TSO-pf has more commit stalls due to unfinished loads in benchmark lu\_ncb.

\noindent\textbf{WMM versus Alpha:}
Figure~\ref{fig: early st} shows the average number of cycles that a store in Alpha can commit before it reaches the commit slot.
However, the early commit (\ie, Ld-St reordering) does not make Alpha outperform WMM (see Figure~\ref{fig: splash retire BW}), because store buffers can already hide the store miss latency.
Note that ROB is typically implemented as a FIFO (\ie, a circular buffer) for register renaming (\eg{}, freeing physical registers in order), precise exceptions, \etc{}
Thus, if the early committed store is in the middle of ROB, its ROB entry cannot be reused by a newly fetched instruction, \ie{}, the effective size of the ROB will not increase.
In summary, the Ld-St reordering in Alpha does not increase performance but complicates the definition (Section~\ref{sec: wmm axiom}).

\begin{figure}[!htb]
    \centering
    \includegraphics[width=\columnwidth]{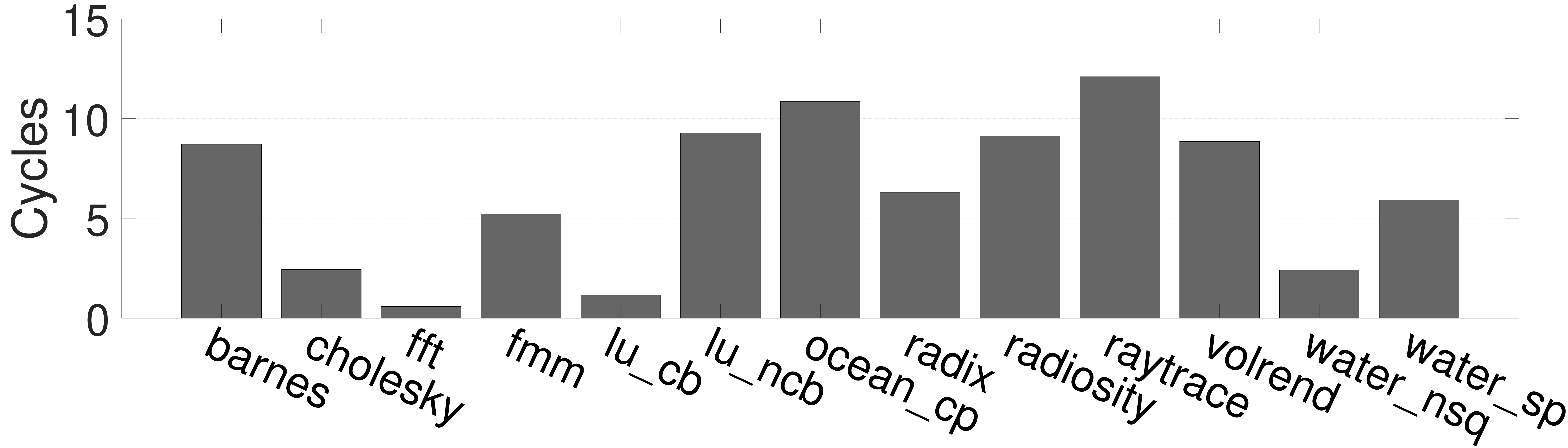}
    \caption{Average cycles to commit stores early in Alpha}\label{fig: early st}
\end{figure}

\section{\wmmSSB{} Model} \label{sec: non atomic mem}

Unlike the multi-copy atomic stores in WMM, stores in some processors (\eg, POWER) are non-atomic due to shared write-through caches or shared store buffers.
If multiple processors share a store buffer or write-through cache, a store by any of these processors may be seen by all these processors before other processors.
Although we could tag stores with processor IDs in the store buffer, it is infeasible to separate values stored by different processors in a cache.

In this section, we introduce a new \IIE{} model, \wmmSSB, which captures the non-atomic store behaviors in a way independent from the sharing topology.
\wmmSSB{} is derived from WMM by adding a new background operation.
We will show later in Section~\ref{sec: ssb impl} why \wmmSSB{} can be implemented using memory systems with non-atomic stores.

\subsection{I\textsuperscript{2}E Definition of \wmmSSB} \label{sec: st propagate}

The structure of the abstract machine of \wmmSSB{} is the same as that of WMM.
To model non-atomicity of stores, \ie, to make a store by one processor readable by another processor before the store updates the \regfile{} memory, \wmmSSB{} introduces a new background operation that copies a store from one store buffer into another.
However, we need to ensure that all stores for an address can still be put in a total order (i.e., the coherence order), and the order seen by any processor is consistent with this total order (\ie, per-location SC).

To identify all the copies of a store in various store buffers, we assign a unique tag $t$ when a store is executed (by being inserted into $sb$), and this tag is copied when a store is copied from one store buffer to another.
When a background operation dequeues a store from a store buffer to the memory, all its copies must be deleted from all the store buffers which have them. 
This requires that all copies of the store are the oldest for that address in their respective store buffers.

All the stores for an address in a store buffer can be strictly ordered as a list, where the \emph{youngest} store is the one that entered the store buffer \emph{last}. 
We make sure that all ordered lists (of all store buffers) can be combined transitively to form a partial order (\ie, no cycle), which has now to be understood in terms of the tags on stores because of the copies.
We refer to this partial order as the \emph{partial coherence order} ($\coOrd$), because it is consistent with the coherence order.

Consider the states of store buffers shown in Figure \ref{fig: store propagate example} (primes are copies).
$A$, $B$, $C$ and $D$ are different stores to the same address, and their tags are $t_A$, $t_B$, $t_C$ and $t_D$, respectively.
$A'$ and $B'$ are copies of $A$ and $B$ respectively created by the background copy operation.
Ignoring $C'$, the partial coherence order contains: $t_D \coOrd t_B \coOrd t_A$ ($D$ is older than $B$, and $B$ is older than $A'$ in P2), and $t_C \coOrd t_B$ ($C$ is older than $B'$ in P3).
Note that $t_D$ and $t_C$ are not related here.

At this point, if we copied $C$ in P3 as $C'$ into P1, we would add a new edge $t_A \coOrd t_C$, breaking the partial order by introducing the cycle $t_A \coOrd t_C \coOrd t_B \coOrd t_A$.
Thus copying of $C$ into P1 should be forbidden in this state.
Similarly, copying a store with tag $t_A$ into P1 or P2 should be forbidden because it would immediately create a cycle: $t_A \coOrd t_A$.
In general, the background copy operation must be constrained so that the partial coherence order is still acyclic after copying.

\begin{figure}[!htb]
	\centering
	\includegraphics[width=0.7\columnwidth]{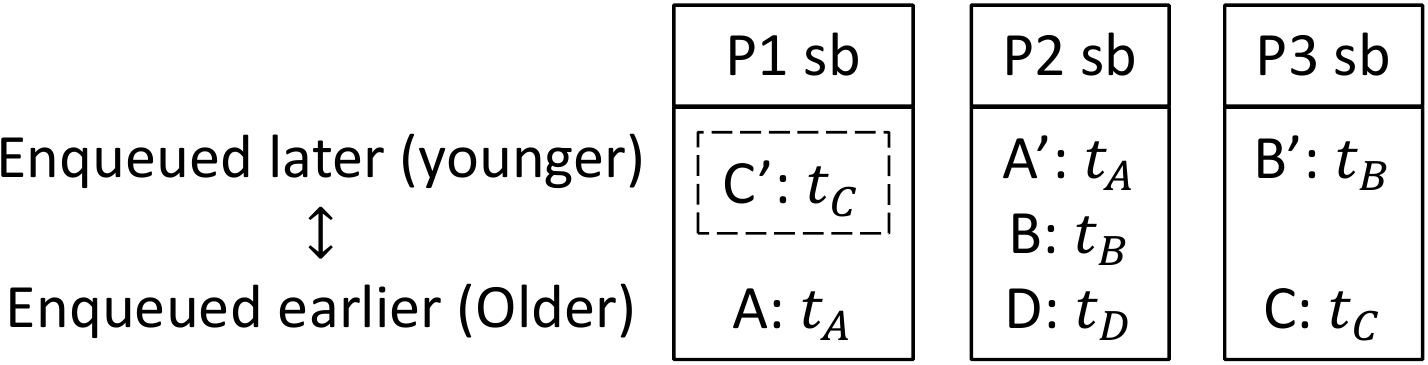}
	\caption{Example states of store buffers} \label{fig: store propagate example}
\end{figure}

Figure~\ref{fig: wmm-s op} shows the background operations of the \wmmSSB{} abstract machine.
The operations that execute instructions in \wmmSSB{} are the same as those in WMM, so we do not show them again.
(The store execution operation in \wmmSSB{} needs to also insert the tag of the store into $sb$).

\begin{figure}[!htb]
    \begin{boxedminipage}{\columnwidth}
        \small
        \textbf{\wmmDeqSbRule} (background store buffer dequeue) \\
        \emph{Predicate:} There is a store $S$ in a store buffer, and all copies of $S$ are the oldest store for that address in their respective store buffers.\\
        \emph{Action:} Assume the $\langle \mathrm{address, value, tag}\rangle$ tuple of store $S$ is $\langle a,v,t\rangle$.
        First, the stale $\langle \mathrm{address, value}\rangle$ pair $\langle a, m[a]\rangle$ is inserted to the $ib$ of every processor whose $sb$ does not contain $a$.
        Then all copies of $S$ are removed from their respective store buffers, and the \regfile{} memory $m[a]$ is updated to $v$.
        \\
        \textbf{\wmmSSBPropSt} (background store copy) \\
        \emph{Predicate:} There is a store $S$ that is in the $sb$ of some processor $i$ but not in the $sb$ of some other processor $j$.
        Additionally, the partial coherence order will still be acyclic if we insert a copy of $S$ into the $sb$ of processor $j$.\\
        \emph{Action:} Insert a copy of $S$ into the $sb$ of processor $j$, and remove all values for the store address of $S$ from the $ib$ of processor $j$.
    \end{boxedminipage}
    \caption{Background operations of \wmmSSB{}}\label{fig: wmm-s op}
\end{figure}

\noindent\textbf{Binding background copy with load execution:}
If the \wmmSSBPropSt{} operation is restricted to always happen right before a load execution operation that reads from the newly created copy, it is not difficult to prove that the \wmmSSB{} model remains the same, \ie, legal behaviors do not change.
In the rest of the paper, we will only consider this ``restricted'' version of \wmmSSB{}.
In particular,  all \wmmSSBPropSt{} operations in the following analysis of litmus tests fall into this pattern.

\subsection{Properties of \wmmSSB}

\wmmSSB{} enforces per-location SC (Figure~\ref{fig: corr}), because it prevents cycles in the order of stores to the same address.
It also allows the same instruction reorderings as WMM does (Figure~\ref{fig: reorder litmus}).
We focus on the store non-atomicity of \wmmSSB.

\noindent\textbf{Non-atomic stores and cumulative fences:}
Consider the litmus tests for non-atomic stores in Figures~\ref{fig: wrc}, \ref{fig: wwc} and \ref{fig: iriw} ($\LdLdFence$ should be $\RecInst$ in these tests).
\wmmSSB{} allows the behavior in Figure~\ref{fig: wrc} by copying $I_1$ into the $sb$ of P2 and then executing $I_2,I_3,I_4,I_5,I_6$ sequentially.
$I_1$ will not be dequeued from $sb$ until $I_6$ returns value 0.
To forbid this behavior, a $\ComInst$ is required between $I_2$ and $I_3$ in P2 to push $I_1$ into memory.
Similarly, \wmmSSB{} allows the behavior in Figure~\ref{fig: wwc} (\ie{}, we copy $I_1$ into the $sb$ of P2 to satisfy $I_2$, and $I_1$ is dequeued after $I_5$ has updated the \regfile{} memory), and we need a $\ComInst$ between $I_2$ and $I_3$ to forbid the behavior.
In both litmus tests, the inserted fences have a cumulative global effect in ordering $I_1$ before $I_3$ and the last instruction in P3.

\wmmSSB{} also allows the behavior in Figure~\ref{fig: iriw} by copying $I_1$ into the $sb$ of P2 to satisfy $I_2$, and copying $I_5$ into the $sb$ of P4 to satisfy $I_6$.
To forbid the behavior, we need to add a $\ComInst$ right after the first load in P2 and P4 (but before the $\LdLdFence/\RecInst$ that we added to stop Ld-Ld reordering).
As we can see, $\ComInst$ and $\RecInst$ are similar to release and acquire respectively. 
Cumulation is achieved by globally advertising observed stores ($\ComInst$) and preventing later loads from reading stale values ($\RecInst$).

\noindent\textbf{Programming properties:}
\wmmSSB{} is the same as WMM in the properties described in Section~\ref{sec: wmm c++},
including the compilation of C++ primitives, maintaining SC for well-synchronized programs, and the conservative way of inserting fences.

\section{\wmmSSB{} Implementations} \label{sec: ssb impl}



Since \wmmSSB{} is strictly more relaxed than WMM, any WMM implementation is a valid \wmmSSB{} implementation.
However, we are more interested in implementations with non-atomic memory systems.
Instead of discussing each specific system one by one, we explain how \wmmSSB{} can be implemented using the ARMv8 flowing model, which is a general abstraction of non-atomic memory systems~\cite{flur2016modelling}.
We first describe the adapted flowing model (\flowModelName) which uses fences in \wmmSSB{} instead of ARM, and then explain how it obeys WMM-S.

\subsection{The Flowing Model (\flowModelName)}

FM consists of a tree of segments $s[i]$ rooted at the \regfile{} memory $m$.
For example, Figure~\ref{fig: OOO+FM} shows four OOO processors (P1$\ldots$P4) connected to a 4-ported FM which has six segments ($s[1\ldots 6]$).
Each segment is a list of memory requests, (\eg, the list of blue nodes in $s[6]$, whose head is at the bottom and the tail is at the top).


OOO interacts with \flowModelName{} in a slightly different way than \ccmModelName.
Every memory request from a processor is appended to the tail of the list of the segment connected to the processor (\eg, $s[1]$ for P1).
OOO no longer contains a store buffer; after a store is committed from ROB, it is directly sent to \flowModelName{} and there is no store response.
When a $\ComInst$ fence reaches the commit slot of ROB, the processor sends a $\ComInst$ request to \flowModelName, and the ROB will not commit the $\ComInst$ fence until \flowModelName{} sends back the response for the $\ComInst$ request.

Inside \flowModelName, there are three background operations:
(1) Two requests in the same segment can be reordered in certain cases;
(2) A load can bypass from a store in the same segment;
(3) The request at the head of the list of a segment can flow into the parent segment (\eg, flow from $s[1]$ into $s[5]$) or the \regfile{} memory (in case the parent of the segment, \eg, $s[6]$, is $m$).
Details of these operations are shown in Figure~\ref{fig: fm-rule}.

\begin{figure}[!htb]
    \begin{boxedminipage}{\columnwidth}
        \small
        \textbf{\flowReorderRule} (reorder memory requests) \\
        \emph{Predicate:} The list of segment $s[i]$ contains two consecutive requests $r_{new}$ and $r_{old}$ ($r_{new}$ is above $r_{old}$ in $s[i]$); and \emph{neither} of the following is true:
        \begin{enumerate}
            \item $r_{new}$ and $r_{old}$ are memory accesses to the same address.
            \item $r_{new}$ is a $\ComInst$ and $r_{old}$ is a store.
        \end{enumerate}
        \emph{Action:} Reorder $r_{new}$ and $r_{old}$ in the list of $s[i]$.
        \\
        \textbf{\flowBypassRule} (store forwarding) \\
        \emph{Predicate:}  The list of segment $s[i]$ contains two consecutive requests $r_{new}$ and $r_{old}$ ($r_{new}$ is above $r_{old}$ in $s[i]$).
        $r_{new}$ is a load, $r_{old}$ is a store, and they are for the same address.\\
        \emph{Action:} we send the load response for $r_{new}$ using the store value of $r_{old}$, and remove $r_{new}$ from the segment.
        \\
        \textbf{\flowFlowRule} (flow request) \\
        \emph{Predicate:} A segment $s[i]$ is not empty. \\
        \emph{Action:} Remove the request $r$ which is the head of the list of $s[i]$.
        If the parent of $s[i]$ in the tree structure is another segment $s[j]$, we append $r$ to the tail of the list of $s[j]$.
        Otherwise, the parent of $s[i]$ is $m$, and we take the following actions according to the type of $r$:
        \begin{itemize}
            \item If $r$ is a load, we send a load response using the value in $m$.
            \item If $r$ is a store $\langle a,v\rangle$, we update $m[a]$ to $v$.
            \item If $r$ is a $\ComInst$, we send a response to the requesting processor and the $\ComInst$ fence can then be committed from ROB.
        \end{itemize}
    \end{boxedminipage}
    \caption{Background operations of \flowModelName}\label{fig: fm-rule}
\end{figure}

It is easy to see that \flowModelName{} abstracts non-atomic memory systems, \eg, Figure~\ref{fig: OOO+FM} abstracts a system in which P1 and P2 share a write-through cache while P3 and P4 share another.

\noindent\textbf{Two properties of \flowModelName+OOO:}
First, \flowModelName+OOO enforces per-location SC because the segments in \flowModelName{} never reorder requests to the same address.
Second, stores for the same address, which lie on the path from a processor to $m$ in the tree structure of \flowModelName, are strictly ordered based on their distance to the tree root $m$; and the combination of all such orderings will not contain any cycle.
For example, in Figure~\ref{fig: OOO+FM}, stores in segments $s[3]$ and $s[6]$ are on the path from P3 to $m$; a store in $s[6]$ is older than any store (for the same address) in $s[3]$, and stores (for the same address) in the same segment are ordered from bottom to top (bottom is older).

\subsection{Relating \flowModelName+OOO to \wmmSSB}\label{sec: ssb hw = model}

\wmmSSB{} can capture the behaviors of any program execution in implementation \flowModelName+OOO in almost the same way that WMM captures the behaviors of \ccmModelName+OOO.
When a store updates the \regfile{} memory in \flowModelName+OOO, \wmmSSB{} performs a \wmmSSBDeqSbRule{} operation to dequeue that store from store buffers to memory.
When an instruction is committed from a ROB in \flowModelName+OOO, \wmmSSB{} executes that instruction.
The following invariants hold after each operation in \flowModelName+OOO and
the corresponding operation in \wmmSSB{}:
\begin{enumerate}
    \item For each instruction committed in \flowModelName+OOO, the execution results in \flowModelName+OOO and \wmmSSB{} are the same.
    \item The \regfile{} memories in \flowModelName+OOO and \wmmSSB{} match.
    \item The $sb$ of each processor $Pi$ in \wmmSSB{} holds exactly all the stores in \flowModelName+OOO that is \emph{observed by the commits of $Pi$} but have not updated the \regfile{} memory.
    (A store is observed by the commits of $Pi$ if it has been either committed by $Pi$ or returned by a load that has been committed by $Pi$).
    \item The order of stores for the same address in the $sb$ of any processor in \wmmSSB{} is exactly the order of those stores on the path from $Pi$ to $m$ in \flowModelName+OOO.
\end{enumerate}
It is easy to see how the invariants are maintained when the \regfile{} memory is updated or a non-load instruction is committed in \flowModelName+OOO.
To understand the commit of a load $L$ to address $a$ with result $v$ in processor $Pi$ in \flowModelName+OOO, we still consider where $\langle a,v \rangle$ resides when $L$ commits.
Similar to WMM, reading \regfile{} memory $m$ or local $ib$ in the load execution operation of \wmmSSB{} covers the cases that $\langle a,v\rangle$ is still in the \regfile{} memory of \flowModelName{} or has already been overwritten by another store in the \regfile{} memory of \flowModelName, respectively.
In case $\langle a, v \rangle$ is a store that has not yet updated the \regfile{} memory in \flowModelName, $\langle a, v \rangle$ must be on the path from $Pi$ to $m$.
In this case, if $\langle a,v\rangle$ has been observed by the commits of $Pi$ before $L$ is committed, then $L$ can be executed by reading the local $sb$ in \wmmSSB{}.
Otherwise, on the path from $Pi$ to $m$, $\langle a,v \rangle$ must be younger than any other store observed by the commits of $Pi$.
Thus, \wmmSSB{} can copy $\langle a, v\rangle$ into the $sb$ of $Pi$ without breaking any invariant.
The copy will not create any cycle in $\coOrd$ because of invariants 3 and 4 as well as the second property of \flowModelName+OOO mentioned above.
After the copy, \wmmSSB{} can have $L$ read $v$ from the local $sb$.

\noindent\textbf{Performance comparison with ARM and POWER:}
As we have shown that \wmmSSB{} can be implemented using the generalized memory system of ARM, we can turn an ARM multicore into a \wmmSSB{} implementation by stopping Ld-St reordering in the ROB.
Since Section~\ref{sec: ld-st reorder} already shows that Ld-St reordering does not affect performance, we can conclude qualitatively that there is  no discernible performance difference between \wmmSSB{} and ARM implementations.
The same arguments apply to the comparison against POWER and RC.

\begin{figure}[!htb]
    \centering
    \begin{minipage}[b]{0.47\columnwidth}
        \includegraphics[width=\textwidth]{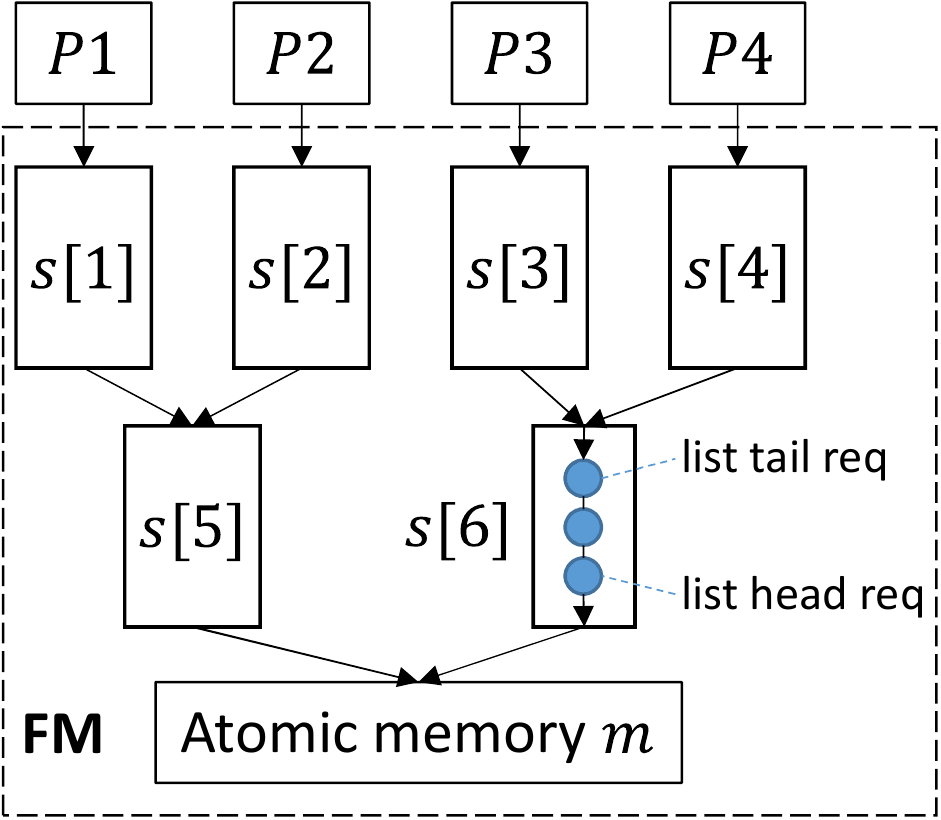}
        \vspace{-20pt}
        \caption{OOO+FM} \label{fig: OOO+FM}
    \end{minipage}
    \hspace{2pt}
    \begin{minipage}[b]{0.48\columnwidth}
        \centering
        \footnotesize
        \begin{tabular}{|l|l|}
            \hline
            Proc. P1 & Proc P2 \\
            \hline
            $\!\!\! I_1: \StInst\ a\ 1 \!\!\!$ & $\!\!\! I_4: r_1=\LdInst\ b \!\!\!$ \\
            $\!\!\! I_2: \MBInst \!\!\!$       & $\!\!\! I_5: \mathsf{if}(r_1 \!\neq\! 1)\ \mathsf{exit} \!\!\!$ \\
            $\!\!\! I_3: \StInst\ b\ 1 \!\!\!$ & $\!\!\! I_6: \StInst\ c\ 1 \!\!\!$ \\
                                               & $\!\!\! I_7: r_2=\LdInst\ c \!\!\!$ \\
                                               & $\!\!\! I_8: r_3 \!=\! a \!+\! r_2 \!-\! 1 \!\!\!$ \\
                                               & $\!\!\! I_9: r_4=\LdInst\ r_3 \!\!\!$ \\
            \hline
            \multicolumn{2}{|l|}{RMO forbids: $r_1=1,\ r_2=1$} \\
            \multicolumn{2}{|l|}{$r_3=a,\ r_4=0$} \\
            \hline
        \end{tabular}
        \caption{RMO dependency order} \label{fig: rmo dep order}
    \end{minipage}
\end{figure}

\section{Problems of RC and RMO} \label{sec: alpha rmo rc bad}

Here we elaborate the problems of RC (both $\mathrm{RC_{sc}}$ and $\mathrm{RC_{pc}}$) and RMO, which have been pointed out in Section~\ref{sec: intro}.

\noindent\textbf{RC:}
Although the RC definition~\cite{gharachorloo1990memory} allows the behaviors of WRC and IRIW (Figures~\ref{fig: wrc} and \ref{fig: iriw}), it disallows the behavior of WWC (Figure~\ref{fig: wwc}).
In WWC, when $I_2$ reads the value of store $I_1$, the RC definition says that $I_1$ is performed with respect to (w.r.t) P2.
Since store $I_5$ has not been issued due to the data dependencies in P2 and P3, $I_1$ must be performed w.r.t P2 before $I_5$.
The RC definition says that ``all writes to the same location are serialized in some order and are performed in that order with respect to any processor''~\cite[Section 2]{gharachorloo1990memory}.
Thus, $I_1$ is before $I_5$ in the serialization order of stores for address $a$, and the final memory value of $a$ cannot be 2 (the value of $I_1$), \ie, RC forbids the behavior of WWC and thus forbids shared write-through caches in implementations.

\noindent\textbf{RMO:}
The RMO definition~\cite[Section D]{weaver1994sparc} is incorrect in enforcing dependency ordering.
Consider the litmus test in Figure~\ref{fig: rmo dep order} ($\MBInst$ is the fence in RMO).
In P2, the execution of $I_6$ is conditional on the result of $I_4$, $I_7$ loads from the address that $I_6$ stores to, and $I_9$ uses the results of $I_7$.
According the definition of dependency ordering in RMO~\cite[Section D.3.3]{weaver1994sparc}, $I_9$ depends on $I_4$ transitively.
Then the RMO axioms~\cite[Section D.4]{weaver1994sparc} dictate that $I_9$ must be after $I_4$ in the memory order, and thus forbid the behavior in Figure~\ref{fig: rmo dep order}.
However, this behavior is possible in hardware with speculative load execution and store forwarding, \ie, $I_7$ first speculatively bypasses from $I_6$, and then $I_9$ executes speculatively to get 0.
Since most architects will not be willing to give up on these two optimizations, \RISCV{} should not adopt RMO.

\section{Conclusion} \label{sec: conclude}

We have proposed two weak memory models, WMM and \wmmSSB{}, for \RISCV{} with different tradeoffs between definitional simplicity and implementation flexibility.
However RISC-V can have only one memory model.
Since there is no obvious evidence that restricting to multi-copy atomic stores affects performance or increases hardware complexity, RISC-V should adopt WMM in favor of simplicity.

\section{Acknowledgment}
We thank all the anonymous reviewers on the different versions of this paper over the two years.
We have also benefited from the discussions with Andy Wright, Thomas Bourgeat, Joonwon Choi, Xiangyao Yu, and Guowei Zhang.
This work was done as part of the Proteus project under the DARPA BRASS Program (grant number 6933218).

\appendices
\section{Proof of Equivalence between WMM $\IIE$ Model and WMM Axiomatic Model}\label{sec:proof}

\newcommand{\MOIIE}{<_{mo\Hyphen{}i2e}}
\newcommand{\POIIE}{<_{po\Hyphen{}i2e}}
\newcommand{\RFIIE}{\xrightarrow{r\!f\Hyphen{}i2e}}
\newcommand{\ExOrd}[1]{<_{e\Hyphen{}#1}}

\newcommand{\SrcSt}{\mathsf{source}}
\newcommand{\OverwriteSt}{\mathsf{overwrite}}
\newcommand{\Addr}{\mathsf{addr}}

Here we present the equivalence proof for the $\IIE$ definition and the axiomatic definition of WMM.

\begin{theorem}[Soundness]\label{thm:wmm:i2e<axi}
    WMM $\IIE{}$ model $\subseteq$ WMM axiomatic model.
\end{theorem}
\begin{proof}
The goal is that for any execution in the WMM $\IIE$ model, we can construct relations $\langle \ProgOrd, \MemOrd, \ReadFrom \rangle$ that have the same program behavior and satisfy the WMM axioms.
To do this, we first introduce the following ghost states to the $\IIE$ model:
\begin{itemize}
    \item Field $\SrcSt$ in the \regfile{} memory: For each address $a$, we add state $m[a].\SrcSt$ to record the store that writes the current memory value.
    \item Fields $\SrcSt$ and $\OverwriteSt$ in the invalidation buffer: For each stale value $\langle a, v\rangle$ in an invalidation buffer, we add state $v.\SrcSt$ to denote the store of this stale value, and add state $v.\OverwriteSt$ to denote the store that overwrites $v.\SrcSt$ in the memory.
    \item Per-processor list $\POIIE$: For each processor, $\POIIE$ is the list of all the instructions that has been executed by the processor.
    The order in $\POIIE$ is the same as the execution order in the processor.
    We also use $\POIIE$ to represent the ordering relation in the list (the head of the list is the oldest/minimum in $\POIIE$).
    \item Global list $\MOIIE$: $\MOIIE$ is a list of all the executed loads, executed fences, and stores that have been dequeued from the store buffers.
    $\MOIIE$ contains instructions from all processors.
    We also use $\MOIIE$ to represent the ordering relation in the list (the head of the list is the oldest/minimum in $\MOIIE$).
    \item Read-from relations $\RFIIE$: $\RFIIE$ is a set of edges.
    Each edge points from a store to a load, indicating that the load had read from the store in the $\IIE$ model.
\end{itemize}
$m[a].\SrcSt$ initially points to the initialization store, and $\POIIE,\MOIIE,\RFIIE$ are all initially empty.
We now show how these states are updated in the operations of the WMM $\IIE$ model.
\begin{enumerate}
    \item \wmmNmRule, \wmmComRule, \wmmRecRule, \wmmStRule: Assume the operation executes an instruction $I$ in processor $i$.
    We append $I$ to the tail of list $\POIIE$ of processor $i$.
    If $I$ is a fence (i.e., the operation is \wmmComRule{} or \wmmRecRule), then we also append $I$ to the tail of list $\MOIIE$.
    \item \wmmDeqSbRule: Assume the operation dequeues a store $S$ for address $a$.
    In this case, we update $m[a].\SrcSt$ to be $S$.
    Let $S_0$ be the original $m[a].\SrcSt$ before this operation is performed.
    Then for each new stale value $\langle a, v\rangle$ inserted into any invalidation buffer, we set $v.\SrcSt = S_0$ and $v.\OverwriteSt = S$.
    We also append $S$ to the tail of list $\MOIIE$.
    \item \wmmLdRule: Assume the operation executes a load $L$ for address $a$ in processor $i$.
    We append $L$  to the tail of list $\POIIE$ of processor $i$.
    The remaining actions depends on how $L$ gets its value in this operation:
    \begin{itemize}
        \item If $L$ reads from a store $S$ in the local store buffer, then we add edge $S\RFIIE L$, and append $L$ to the tail of list $\MOIIE$.
        \item If $L$ reads the \regfile{} memory $m[a]$, then we add edge $m[a].\SrcSt \RFIIE L$, and append $L$ to the tail of list $\MOIIE$.
        \item If $L$ reads a stale value $\langle a, v\rangle$ in the local invalidation buffer, then we add edge $v.\SrcSt \RFIIE L$, and we insert $L$ to be right before $v.\OverwriteSt$ in list $\MOIIE$ (i.e., $L$ is older than $v.\OverwriteSt$, but is younger than any other instruction which is older than $v.\OverwriteSt$).
    \end{itemize} 
\end{enumerate}
As we will see later, at the end of the $\IIE$ execution, $\POIIE$, $\MOIIE$ and $\RFIIE$ will become the $\langle \ProgOrd, \MemOrd, \ReadFrom \rangle$ relations that satisfy the WMM axioms.
Before getting there, we show that the $\IIE$ model has the following invariants after each operation is performed:
\begin{enumerate}
    \item \label{inv:sound:mem-val} For each address $a$, $m[a].\SrcSt$ in the $\IIE$ model is the youngest store for $a$ in $\MOIIE$.
    \item \label{inv:sound:ldf} All loads and fences that have been executed in the $\IIE$ model are in $\MOIIE$.
    \item \label{inv:sound:sb} An executed store is either in $\MOIIE$ or in store buffer, i.e., for each processor $i$, the store buffer of processor $i$ contains exactly every store that has been executed in the $\IIE$ model but is not in $\MOIIE$.
    \item \label{inv:sound:sb-ord} For any two stores $S_1$ and $S_2$ for the same address in the store buffer of any processor $i$ in the $\IIE$ model, if $S_1$ is older than $S_2$ in the store buffer, then $S_1\POIIE S_2$.
    \item \label{inv:sound:sb-ib} For any processor $i$ and any address $a$, address $a$ cannot be present in the store buffer and invalidation buffer of processor $i$ at the same time.
    \item \label{inv:sound:ib} For any stale value $v$ for any address $a$ in the invalidation buffer of any processor $i$ in the $\IIE$ model, the following invariants hold:
    \begin{enumerate}
        \item \label{inv:sound:ib:ovw} $v.\SrcSt$ and $v.\OverwriteSt$ are in $\MOIIE$, and $v.\SrcSt \MOIIE v.\OverwriteSt$, and there is no other store for $a$ between them in $\MOIIE$.
        \item \label{inv:sound:ib:rec} For any $\RecInst$ fence $F$ that has been executed by processor $i$ in the $\IIE$ model, $F \MOIIE v.\OverwriteSt$.
        \item \label{inv:sound:ib:st} For any store $S$ for $a$ that has been executed by processor $i$ in the $\IIE$ model, $S \MOIIE v.\OverwriteSt$.
        \item \label{inv:sound:ib:ld} For any load $L$ for $a$ that has been executed by processor $i$ in the $\IIE$ model, if store $S \RFIIE L$, then $S \MOIIE v.\OverwriteSt$.
    \end{enumerate}
    \item \label{inv:sound:ib-ord} For any two stale values $v_1$ and $v_2$ for the same address in the invalidation buffer of any processor $i$ in the $\IIE$ model, if $v_1$ is older than $v_2$ in the invalidation buffer, then $v_1.\SrcSt \MOIIE v_2.\SrcSt$.
    \item \label{inv:sound:inst-ord} For any instructions $I_1$ and $I_2$, if $I_1\POIIE I_2$ and $\orderFunc(I_1, I_2)$ and $I_2$ is in $\MOIIE$, then $I_1 \MOIIE I_2$. 
    \item \label{inv:sound:rf} For any load $L$ and store $S$, if $S\RFIIE L$, then the following invariants hold:
    \begin{enumerate}
        \item \label{inv:sound:rf:sb} If $S$ not in $\MOIIE$, then $S$ is in the store buffer of the processor of $L$, and $S\POIIE L$, and there is no store $S'$ for the same address in the same store buffer such that $S\POIIE S'\POIIE L$.
        \item \label{inv:sound:rf:mem} If $S$ is in $\MOIIE$, then $S = \max_{mo\Hyphen{}i2e}\{ S'\ |\ S'.\Addr = L.\Addr\ \wedge\ (S'\POIIE L\ \vee\ S'\MOIIE L) \}$, and there is no other store $S''$ for the same address in the store buffer of the processor of $L$ such that $S''\POIIE L$.
    \end{enumerate}
\end{enumerate}
We now prove inductively that all invariants hold after each operation $R$ is performed in the $\IIE$ model, i.e., we assume all invariants hold before $R$ is performed.
In case performing $R$ changes some states (e.g., $\MOIIE$), we use superscript 0 to denote the state before $R$ is performed (e.g., $\MOIIE^0$) and use superscript 1 to denote the state after $R$ is performed (e.g., $\MOIIE^1$).
Now we consider the type of $R$:
\begin{enumerate}
    \item \wmmNmRule: All invariants still hold.
    
    \item \wmmStRule: Assume $R$ executes a store $S$ for address $a$ in processor $i$.
    $R$ changes the states of the store buffer, invalidation buffer, and $\POIIE$ of processor $i$.
    Now we consider each invariant.
    \begin{itemize}
        \item Invariant~\ref{inv:sound:mem-val}, \ref{inv:sound:ldf}: These are not affected.
        
        \item Invariant~\ref{inv:sound:sb}: This invariant still holds for the newly executed store $S$.
        
        \item Invariant~\ref{inv:sound:sb-ord}: Since $S$ becomes the youngest store in the store buffer of processor $i$, this invariant still holds.
        
        \item Invariant~\ref{inv:sound:sb-ib}: Since $R$ will clear address $a$ from the invalidation buffer of processor $i$, this invariant still holds.
        
        \item Invariant~\ref{inv:sound:ib}: Invariants~\ref{inv:sound:ib:ovw}, \ref{inv:sound:ib:rec}, \ref{inv:sound:ib:ld} are not affected.
        Invariant~\ref{inv:sound:ib:st} still holds because there is no stale value for $a$ in the invalidation buffer of processor $i$ after $R$ is performed.
        
        \item Invariant~\ref{inv:sound:ib-ord}: This is not affected, because $R$ can only remove values from the invalidation buffer.
        
        \item Invariant~\ref{inv:sound:inst-ord}: This is not affected because $R$ is not in $\MOIIE$.
        
        \item Invariant~\ref{inv:sound:rf}: Consider load $L^*$ and store $S^*$ for address $a$ such that $S^*\RFIIE L^*$ and $L^*$ is from processor $i$.
        We need to show that this invariant still holds for $L^*$ and $S^*$.
        Since $L^*$ has been executed, we have $L^*\POIIE^1 S$.
        Thus this invariant cannot be affected.
    \end{itemize}
    
    \item \wmmComRule: Assume $R$ executes a $\ComInst$ fence $F$ in processor $i$.
    $R$ adds $F$ to the end of the $\POIIE$ of processor $i$ and adds $F$ to the end of $\MOIIE$.
    Now we consider each invariant.
    \begin{itemize}
        \item Invariants~\ref{inv:sound:mem-val}, \ref{inv:sound:sb}, \ref{inv:sound:sb-ord}, \ref{inv:sound:sb-ib}, \ref{inv:sound:ib}, \ref{inv:sound:ib-ord}, \ref{inv:sound:rf}: These are not affected.
        
        \item Invariant~\ref{inv:sound:ldf}: This still holds because $F$ is added to $\MOIIE$.
        
        \item Invariant~\ref{inv:sound:inst-ord}: Consider instruction $I$ in processor $i$ such that $I\POIIE F$ and $\orderFunc(I, F)$.
        We need to show that $I\MOIIE^1 F$.
        Since $\orderFunc(I, F)$, $I$ can be a load, or store, or fence.
        If $I$ is a load or fence, since $I$ has been executed, invariant~\ref{inv:sound:ldf} says that $I$ is in $\MOIIE^0$ before $R$ is performed.
        Since $F$ is added to the end of $\MOIIE$, $I\MOIIE^1 F$.
        If $I$ is a store, the predicate of $R$ says that $I$ is not in the store buffer.
        Then invariant~\ref{inv:sound:sb} says that $I$ must be in $\MOIIE^0$, and we have $I\MOIIE^1 F$.
    \end{itemize}
    
    \item \wmmRecRule: Assume $R$ executes a $\RecInst$ fence $F$ in processor $i$.
    $R$ adds $F$ to the end of the $\POIIE$ of processor $i$, adds $F$ to the end of $\MOIIE$, and clear the invalidation buffer of processor $i$.
    Now we consider each invariant.
    \begin{itemize}
        \item Invariants~\ref{inv:sound:mem-val}, \ref{inv:sound:sb}, \ref{inv:sound:sb-ord}, \ref{inv:sound:rf}: These are not affected.
        
        \item Invariant~\ref{inv:sound:ldf}: This still holds because $F$ is added to $\MOIIE$.
        
        \item Invariants~\ref{inv:sound:sb-ib}, \ref{inv:sound:ib}, \ref{inv:sound:ib-ord}: These invariant still hold because the invalidation buffer of processor $i$ is now empty.
        
        \item Invariant~\ref{inv:sound:inst-ord}: Consider instruction $I$ in processor $i$ such that $I\POIIE F$ and $\orderFunc(I, F)$.
        We need to show that $I\MOIIE^1 F$.
        Since $\orderFunc(I, F)$, $I$ can be a load or fence.
        Since $I$ has been executed, $I$ must be in $\MOIIE^0$ before $R$ is performed according to invariant~\ref{inv:sound:ldf}.
        Thus, $I\MOIIE^1 F$.
    \end{itemize}
    
    \item \wmmDeqSbRule: Assume $R$ dequeues a store $S$ for address $a$ from the store buffer of processor $i$.
    $R$ changes the store buffer of processor $i$, the \regfile{} memory $m[a]$, and invalidation buffers of other processors.
    $R$ also adds $S$ to the end of $\MOIIE$.
    Now we consider each invariant.
    \begin{itemize}
        \item Invariant~\ref{inv:sound:mem-val}: This invariant still holds, because $m[a].\SrcSt^1=S$ and $S$ becomes the youngest store for $a$ in $\MOIIE^1$.
        
        \item Invariant~\ref{inv:sound:ldf}: This is not affected.
        
        \item Invariant~\ref{inv:sound:sb}: This invariant still holds, because $S$ is removed from store buffer and added to $\MOIIE$.
        
        \item Invariants~\ref{inv:sound:sb-ord}: This is not affected because we only remove stores from the store buffer.
        
        \item Invariant~\ref{inv:sound:sb-ib}: The store buffer and invaliation buffer of processor $i$ cannot be affected.
        The store buffer and invalidation buffer of processor $j$ ($\neq i$) may be affected, because $m[a]^0$ may be inserted into the invalidation buffer of processor $j$.
        The predicate of $R$ ensures that the insertion will not happen if the store buffer of processor $j$ contains address $a$, so the invariant still holds.
        
        \item Invariant~\ref{inv:sound:ib}: We need to consider the influence on both existing stale values and the newly inserted stale values.
        \begin{enumerate}
            \item Consider stale value $\langle a, v\rangle$ which is in the invalidation buffer of processor $j$ both before and after operation $R$ is performed.
            This implies $j \neq i$, because the store buffer of processor $i$ contains address $a$ before $R$ is performed, and invariant~\ref{inv:sound:sb-ib} says that the invalidation buffer of processor $i$ cannot have address $a$ before $R$ is performed.
            Now we show that each invariant still holds for $\langle a, v\rangle$.
            \begin{itemize}
                \item Invariant~\ref{inv:sound:ib:ovw}: This still holds because $S$ is the youngest in $\MOIIE$.
                \item Invariant~\ref{inv:sound:ib:rec}: This is not affected.
                \item Invariant~\ref{inv:sound:ib:st}: This is not affected because $S$ is not executed by processor $j$.
                \item Invariant~\ref{inv:sound:ib:ld}: Since $S$ is not in $\MOIIE^0$, invariant~\ref{inv:sound:rf:sb} says that any load that has read $S$ must be from process $i$.
                Since $i\neq j$, this invariant cannot be affected.
            \end{itemize}
            \item Consider the new stale value $\langle a, v\rangle$ inserted to the invalidation buffer of process $j$ ($\neq i$).
            According to \wmmDeqSbRule, $v = m[a]^0$, $v.\SrcSt = m[a].\SrcSt^0$, and $v.\OverwriteSt = S$.
            Now we check each invariant.
            \begin{itemize}
                \item Invariant~\ref{inv:sound:ib:ovw}: According to invariant~\ref{inv:sound:mem-val}, $v.\SrcSt = m[a]^0.\SrcSt$ is the youngest store for $a$ in $\MOIIE^0$.
                Since $S$ (i.e., $v.\OverwriteSt$) is appended to the tail of $\MOIIE^0$, this invariant still holds.
                \item Invariant~\ref{inv:sound:ib:rec}: According to invariant~\ref{inv:sound:ldf}, any $\RecInst$ fence $F$ executed by processor $j$ must be in $\MOIIE^0$.
                Thus, $F\MOIIE^1 S = v.\OverwriteSt$, and the invariant still holds.
                \item Invariant~\ref{inv:sound:ib:st}: The predicate of $R$ says that the store buffer of processor $j$ cannot contain address $a$.
                Therefore, according to invariant~\ref{inv:sound:sb}, any store $S'$ for $a$ executed by processor $j$ must be in $\MOIIE^0$.
                Thus, $S'\MOIIE^1 S = v.\OverwriteSt$, and the invariant still holds.
                \item Invariant~\ref{inv:sound:ib:ld}: Consider load $L$ for address $a$ that has been executed by processor $j$.
                Assume store $S'\RFIIE L$.
                The predicate of $R$ says that the store buffer of processor $j$ cannot contain address $a$.
                Thus, $S'$ must be in $\MOIIE^0$ according to invariant~\ref{inv:sound:rf:sb}.
                Therefore, $S'\MOIIE^1 S = v.\OverwriteSt$, and the invariant still holds.
            \end{itemize}
        \end{enumerate}
        
        \item Invariant~\ref{inv:sound:inst-ord}: Consider instruction $I$ such that $I\POIIE S$ and $\orderFunc(I, S)$.
        Since $\orderFunc(I, S)$, $I$ can be a load, fence, or store for $a$.
        If $I$ is a load or fence, then invariant~\ref{inv:sound:ldf} says that $I$ is in $\MOIIE^0$, and thus $I\MOIIE^1 S$, i.e., the invariant holds.
        If $I$ is a store for $a$, then the predicate of $R$ and invariant~\ref{inv:sound:sb-ord} imply that $I$ is not in the store buffer of processor $i$.
        Then invariant~\ref{inv:sound:sb} says that $I$ must be in $\MOIIE^0$, and thus $I\MOIIE^1 S$, i.e., the invariant holds.
        
        \item Invariant~\ref{inv:sound:rf}: We need to consider the influence on both loads that read $S$ and loads that reads stores other than $S$.
        \begin{enumerate}
            \item Consider load $L$ for address $a$ that reads from $S$, i.e., $S\RFIIE L$.
            Since $S$ is not in $\MOIIE^0$ before $R$ is performed, invariant~\ref{inv:sound:rf:sb} says that $L$ must be executed by processor $i$, $S\POIIE L$, and there is no store $S'$ for $a$ in the store buffer of processor $i$ such that $S\POIIE S'\POIIE L$.
            Now we show that both invariants still hold for $S\RFIIE L$.
            \begin{itemize}
                \item Invariant~\ref{inv:sound:rf:sb}: This is not affected because $S$ is in $\MOIIE^1$ after $R$ is performed.
                \item Invariant~\ref{inv:sound:rf:mem}: Since $S\POIIE L$ and $S$ is the youngest in $\MOIIE^1$, $S$ satisfies the $\max_{mo\Hyphen{}i2e}$ formula.
                We prove the rest of this invariant by contradiction, i.e., we assume there is store $S'$ for $a$ in the store buffer of processor $i$ after $R$ is performed such that $S' \POIIE L$.
                Note that $\POIIE$ is not changed by $R$.
                The predicate of $R$ ensures that $S$ is the oldest store for $a$ in the store buffer.
                Invariant~\ref{inv:sound:sb-ord} says that $S\POIIE S'$.
                Now we have $S\POIIE S'\POIIE L$ (before $R$ is performed), contradicting with invariant~\ref{inv:sound:rf:sb}.
            \end{itemize}
            \item Consider load $L$ for address $a$ from processor $j$ that reads from store $S^*$ ($\neq S$), i.e., $S\neq S^*\RFIIE L$.
            Now we show that both invariants still hold for $S^*\RFIIE L$.
            \begin{itemize}
                \item Invariant~\ref{inv:sound:rf:sb}: This invariant cannot be affected, because performing $R$ can only remove a store from a store buffer.
                \item Invariant~\ref{inv:sound:rf:mem}: This invariant can only be affected when $S^*$ is in $\MOIIE^0$.
                Since $R$ can only remove a store from a store buffer, the second half of this invariant is not affected (i.e, no store $S''$ in the store buffer and so on).
                We only need to focus on the $\max_{mo\Hyphen{}i2e}$ formula, i.e., $S^* = \max_{mo\Hyphen{}i2e}\{ S'\ |\ S'.\Addr = a\ \wedge\ (S'\POIIE L\ \vee\ S'\MOIIE L) \}$.
                Since $L\MOIIE^1 S$, this formula can only be affected when $S\POIIE L$ and $i = j$.
                In this case, before $R$ is performed, $S$ is in the store buffer of processor $i$, and $S\POIIE L$, and $L$ reads from $S^*\neq S$.
                This contradicts with invariant~\ref{inv:sound:rf:mem} which is assume to hold before $R$ is performed.
                Thus, the $\max_{mo\Hyphen{}i2e}$ formula cannot be affected either, i.e., the invariant holds.
            \end{itemize}
        \end{enumerate}
    \end{itemize}
    
    \item \wmmLdRule{} that reads from local store buffer: Assume $R$ executes a load $L$ for address $a$ in processor $i$, and $L$ reads from store $S$ in the local store buffer.
    $R$ appends $L$ to the $\POIIE$ of processor $i$, appends $L$ to $\MOIIE$, and adds $S\RFIIE^1 L$.
    Note that $R$ does not change any invalidation buffer or store buffer.
    Now we consider each invariant.
    \begin{itemize}
        \item Invariants~\ref{inv:sound:mem-val}, \ref{inv:sound:sb}, \ref{inv:sound:sb-ord}, \ref{inv:sound:sb-ib}, \ref{inv:sound:ib-ord}: These are not affected.
        
        \item Invariant~\ref{inv:sound:ldf}: This still holds because $L$ is added to $\MOIIE$.
        
        \item Invariant~\ref{inv:sound:ib}: We consider each invariant.
        \begin{itemize}
            \item Invariants~\ref{inv:sound:ib:ovw}, \ref{inv:sound:ib:rec}, \ref{inv:sound:ib:st}: These are not affected.
            \item Invariant~\ref{inv:sound:ib:ld}: $L$ can only influence stale values for $a$ in the invalidation buffer of processor $i$.
            However, since $S$ is in the store buffer of processor $i$ before $R$ is performed, invariant~\ref{inv:sound:sb-ib} says that the invalidation buffer of processor $i$ cannot contain address $a$.
            Therefore this invariant still holds.
        \end{itemize}
        
        \item Invariant~\ref{inv:sound:inst-ord}: We consider instruction $I$ such that $I\POIIE^1 L$ and $\orderFunc(I, L)$.
        Since $\orderFunc(I, L)$, $I$ can only be a $\RecInst$ fence or a load for $a$.
        In either case, invariant~\ref{inv:sound:ldf} says that $I$ is in $\MOIIE^0$.
        Since $L$ is appended to the end of $\MOIIE$, $I\MOIIE^1 L$, i.e., the invariant still holds.
        
        \item Invariant~\ref{inv:sound:rf}: Since $R$ does not change any store buffer or any load/store already in $\MOIIE^0$, $R$ cannot affect this invariant for loads other than $L$.
        We only need to show that $S\RFIIE^1 L$ satisfies this invariant.
        Since $S$ is in the store buffer, invariant~\ref{inv:sound:sb} says that $S$ is not in $\MOIIE$.
        Therefore we only need to consider invariant~\ref{inv:sound:rf:sb}.
        We prove by contradiction, i.e., we assume there is store $S'$ for $a$ in the store buffer of processor $i$ and $S\POIIE^1 S'\POIIE^1 L$.
        Since $R$ does not change store buffer states, $S$ and $S'$ are both in the store buffer before $R$ is performed.
        We also have $S\POIIE^0 S'$ (because the only change in $\POIIE$ is to append $L$ to the end).
        According to the predicate of $R$, $S$ should be younger than $S'$, so $S'\POIIE^0 S$ (according to invariant~\ref{inv:sound:sb-ord}), contradicting with the previous conclusion.
        Therefore, the invariant still holds.
    \end{itemize}
    
    \item \wmmLdRule{} that reads from \regfile{} memory: Assume $R$ executes a load $L$ for address $a$ in processor $i$, and $L$ reads from \regfile{} memory $m[a]$.
    $R$ appends $L$ to the $\POIIE$ of processor $i$, appends $L$ to $\MOIIE$, adds $m[a].\SrcSt \RFIIE^1 L$, and may remove stale values from the invalidation buffer of processor $i$.
    Now we consider each invariant.
    \begin{itemize}
        \item Invariants~\ref{inv:sound:mem-val}, \ref{inv:sound:sb}, \ref{inv:sound:sb-ord}: These are not affected.
        
        \item Invariant~\ref{inv:sound:ldf}: This still holds because $L$ is added to $\MOIIE$.
        
        \item Invariants~\ref{inv:sound:sb-ib}, \ref{inv:sound:ib-ord}: These are not affected because $R$ only remove values from an invalidation buffer.
        
        \item Invariant~\ref{inv:sound:ib}: We consider each invariant.
        \begin{itemize}
            \item Invariants~\ref{inv:sound:ib:ovw}, \ref{inv:sound:ib:rec}, \ref{inv:sound:ib:st}: These are not affected because $R$ only remove values from an invalidation buffer.
            \item Invariant~\ref{inv:sound:ib:ld}: $L$ can only influence stale values for $a$ in the invalidation buffer of processor $i$.
            However, $R$ will remove address $a$ from the the invalidation buffer of processor $i$.
            Therefore this invariant still holds.
        \end{itemize}
        
        \item Invariant~\ref{inv:sound:inst-ord}: We consider instruction $I$ such that $I\POIIE^1 L$ and $\orderFunc(I, L)$.
        Since $\orderFunc(I, L)$, $I$ can only be a $\RecInst$ fence or a load for $a$.
        In either case, invariant~\ref{inv:sound:ldf} says that $I$ is in $\MOIIE^0$.
        Since $L$ is appended to the end of $\MOIIE$, $I\MOIIE^1 L$, i.e., the invariant still holds.
        
        \item Invariant~\ref{inv:sound:rf}: Since $R$ does not change any store buffer or any load/store already in $\MOIIE^0$, $R$ cannot affect this invariant for loads other than $L$.
        We only need to show that $m[a].\SrcSt\RFIIE^1 L$ satisfies this invariant ($m[a]$ is not changed before and after $R$ is performed).
        According to invariant~\ref{inv:sound:mem-val}, $m[a].\SrcSt$ is the youngest store for $a$ in $\MOIIE^0$.
        Therefore we only need to consider invariant~\ref{inv:sound:rf:mem}.
        Since we also have $m[a].\SrcSt \POIIE^1 L$, $\max_{mo\Hyphen{}i2e}\{ S'\ |\ S'.\Addr = a\ \wedge\ (S'\POIIE L\ \vee\ S'\MOIIE L) \}$ will return $m[a].\SrcSt$, i.e., the first half the invariant holds.
        The predicate of $R$ ensures that there is no store for $a$ in the store buffer of processor $i$, so the second half the invariant also holds.
    \end{itemize}
    
    \item \wmmLdRule{} that reads from the invalidation buffer: Assume $R$ executes a load $L$ for address $a$ in processor $i$, and $L$ reads from the stale value $\langle a, v\rangle$ in the local invalidation buffer.
    $R$ appends $L$ to the $\POIIE$ of processor $i$, appends $L$ to $\MOIIE$, adds $v.\SrcSt \RFIIE^1 L$, and may remove stale values from the invalidation buffer of processor $i$.
    Now we consider each invariant.
    \begin{itemize}
        \item Invariants~\ref{inv:sound:mem-val}, \ref{inv:sound:sb}, \ref{inv:sound:sb-ord}: These are not affected.
        
        \item Invariant~\ref{inv:sound:ldf}: This still holds because $L$ is added to $\MOIIE$.
        
        \item Invariants~\ref{inv:sound:sb-ib}, \ref{inv:sound:ib-ord}: These are not affected because $R$ can only remove values from an invalidation buffer.
        
        \item Invariant~\ref{inv:sound:ib}: We consider each invariant.
        \begin{itemize}
            \item Invariants~\ref{inv:sound:ib:ovw}, \ref{inv:sound:ib:rec}, \ref{inv:sound:ib:st}: These are not affected because $R$ can only remove values from an invalidation buffer.
            \item Invariant~\ref{inv:sound:ib:ld}: Only stale values in the invalidation buffer of processor $i$ can be affected.
            Consider stale value $\langle a, v'\rangle$ in the invalidation buffer of processor $i$ after $R$ is performed.
            We need to show that $v.\SrcSt \MOIIE^1 v'.\OverwriteSt$.
            Since $v'$ is not removed in $R$, $v'$ must be either $v$ or younger than $v$ in the invalidation buffer before $R$ is performed.
            According to invariant~\ref{inv:sound:ib-ord}, either $v'.\SrcSt = v.\SrcSt$ or $v.\SrcSt \MOIIE^0 v'.\SrcSt$.
            Since $v'.\SrcSt \MOIIE^0 v'.\OverwriteSt$ according to invariant~\ref{inv:sound:ib:ovw}, $v.\SrcSt \MOIIE^1 v'.\OverwriteSt$.
        \end{itemize}
        
        \item Invariant~\ref{inv:sound:inst-ord}: We consider instruction $I$ such that $I\POIIE^1 L$ and $\orderFunc(I, L)$, and we need to show that $I\MOIIE^1 L$.
        Since $\orderFunc(I, L)$, $I$ can only be a $\RecInst$ fence or a load for $a$.
        If $I$ is a $\RecInst$ fence, then invariant~\ref{inv:sound:ib:rec} says that $I\MOIIE^0 v.\OverwriteSt$.
        Since we insert $L$ right before $v.\OverwriteSt$, we still have $I\MOIIE^1 L$.
        If $I$ is a load for $a$, then invariant~\ref{inv:sound:ib:ld} says that $I\MOIIE^0 v.\OverwriteSt$, and thus we have $I\MOIIE^1 L$.
        
        \item Invariant~\ref{inv:sound:rf}: Since $R$ does not change any store buffer or any load/store already in $\MOIIE^0$, $R$ cannot affect this invariant for loads other than $L$.
        We only need to show that $v.\SrcSt\RFIIE^1 L$ satisfies this invariant.
        Since $v.\SrcSt$ is in $\MOIIE^0$, we only need to consider invariant~\ref{inv:sound:rf:mem}.
        The predicate of $R$ ensures that the store buffer of processor $i$ cannot contain address $a$, so the second half of the invariant holds (i.e., there is no $S''$ and so on).
        
        Now we prove the first half of the invariant, i.e., consider $\max_{mo\Hyphen{}i2e}\{ S'\ |\ S'.\Addr = a\ \wedge\ (S'\POIIE^1 L\ \vee\ S'\MOIIE^1 L) \}$.
        First note that since $v.\SrcSt \MOIIE^0 v.\OverwriteSt$, $v.\SrcSt\MOIIE^1 L$.
        Thus $v.\SrcSt$ is in set $\{ S'\ |\ S'.\Addr = a\ \wedge\ (S'\POIIE^1 L\ \vee\ S'\MOIIE^1 L) \}$.
        Consider any store $S$ that is also in this set, then $S\POIIE^1 L\ \vee\ S\MOIIE^1 L$ must be true.
        If $S\POIIE^1 L$, $S$ is executed in processor $i$ before $R$ is performed.
        Invariant~\ref{inv:sound:ib:st} says that $S\MOIIE^0 v.\OverwriteSt \Rightarrow S\MOIIE^1 v.\OverwriteSt$.
        If $S\MOIIE^1 L$, then $S\MOIIE^1 L\MOIIE^1 v.\OverwriteSt$.
        In either case, we have $S\MOIIE v.\OverwriteSt$.
        Since we have proved invariant~\ref{inv:sound:ib:ovw} holds after $R$ is performed, either $S=v.\SrcSt$ or $S\MOIIE v.\OverwriteSt$.
        Therefore, $\max_{mo\Hyphen{}i2e}$ will return $v.\SrcSt$.
    \end{itemize}
\end{enumerate}
It is easy to see that at the end of the $\IIE$ execution (of a program), there is no instruction to execute in each processor and all store buffers are empty (i.e., all exected loads stores and fences are in $\MOIIE$).
At that time, if we define axiomatic relations $\ProgOrd,\MemOrd,\ReadFrom$ as $\POIIE,\MOIIE,\RFIIE$ respectively, then invariants~\ref{inv:sound:inst-ord} and \ref{inv:sound:rf:mem} become the \AxiInstOrd{} and \AxiLdVal{} axioms respectively.
That is, $\langle \POIIE,\MOIIE,\RFIIE \rangle$ are the relations that satisfy the WMM axioms and have the same program behavior as the $\IIE$ execution.
\end{proof}

\begin{theorem}[Completeness]\label{thm:wmm:axi<i2e}
    WMM axiomatic model $\subseteq$ WMM \IIE{} model.
\end{theorem}
\begin{proof}
The goal is that for any axiomatic relations $\langle \ProgOrd, \MemOrd, \ReadFrom \rangle$ that satisfy the WMM axioms, we can run the same program in the $\IIE$ model and get the same program behavior.
We will devise an algorithm to operate the $\IIE$ model to get the same program behavior as in axiomatic relations $\langle \ProgOrd, \MemOrd, \ReadFrom \rangle$.
In the algorithm, for each instruction in the $\IIE$ model, we need to find its corresponding instruction in the $\ProgOrd$ in axiomatic relations.
Note that this mapping should be an one-to-one mapping, i.e., one instruction in the $\IIE$ model will exactly correspond to one instruction in the axiomatic relations and vice versa, so we do not distinguish between the directions of the mapping.
The algorithm will create this mapping incrementally.
Initially (i.e., before the $\IIE$ model performs any operation), for each processor $i$, we only map the next instruction to execute in processor $i$ of the $\IIE$ model to the oldest instruction in the $\ProgOrd$ of processor $i$ in the axiomatic relations.
After the algorithm starts to operate the $\IIE$ model, whenever we have executed an instruction in a processor in the $\IIE$ model, we map the next instruction to execute in that processor in the $\IIE$ model to the oldest unmapped instruction in  the $\ProgOrd$ of that processor in the axiomatic relations.
The mapping scheme obviously has the following two properties:
\begin{itemize}
    \item The $k$-th executed instruction in a processor in the $\IIE$ model is mapped to the $k$-th oldest instruction in the $\ProgOrd$ of that processor in the axiomatic relations.
    \item In the $\IIE$ model, when a processor has executed $x$ instructions, only the first $x+1$ instructions (i.e., the executed $x$ instructions and the next instruction to execute) of that processor are mapped to instructions in the axiomatic relations.
\end{itemize}
Of course, later in the proof, we will show that the two corresponding instructions (one in the $\IIE$ model and the other in the axiomatic relations) have the same instruction types, same load/store addresses (if they are memory accesses), same store data (if they are stores), and same execution results.
In the following, we will assume the action of adding new instruction mappings as an implicit procedure in the algorithm, so we do not state it over and over again when we explain the algorithm.
When there is no ambiguity, we do not distinguish an instruction in the $\IIE$ model and an instruction in the axiomatic relations if these two instructions corresponds to each other (i.e., the algorithm has built the mapping between them).

Now we give the details of the algorithm.
The algorithm begins with the $\IIE$ model (in initial state), an empty set $Z$, and a queue $Q$ which contains all the memory and fence instructions in $\MemOrd$.
The order of instructions in $Q$ is the same as $\MemOrd$, i.e., the head of $Q$ is the oldest instruction in $\MemOrd$.
The instructions in $Q$ and $Z$ are all considered as instructions in the axiomatic relations.
In each step of the algorithm, we perform one of the followings actions:
\begin{enumerate}
    \item \label{sim:nm} If the next instruction of some processor in the \IIE{} model is a non-memory instruction, then we perform the \wmmNmRule{} operation to execute it in the $\IIE$ model.
    \item \label{sim:ex-st} Otherwise, if the next instruction of some processor in the \IIE{} model is a store, then we perform the \wmmStRule{} operation to execute that store in the $\IIE$ model.
    \item \label{sim:ld-z} Otherwise, if the next instruction of some processor in the \IIE{} model is mapped to a load $L$ in set $Z$, then we perform the \wmmLdRule{} operation to execute $L$ in the $\IIE$ model, and we remove $L$ from $Z$.
    \item Otherwise, we pop out instruction $I$ from the head of $Q$ and process it in the following way:
    \begin{enumerate}
        \item \label{sim:deq-st} If $I$ is a store, then $I$ must have been mapped to a store in some store buffer (we will prove this), and we perform the \wmmDeqSbRule{} operation to dequeue $I$ from the store buffer in the \IIE{} model.
        \item \label{sim:rec} If $I$ is a $\RecInst$ fence, then $I$ must have been mapped to the next instruction to execute in some processor (we will prove this), and we perform the \wmmRecRule{} operation to execute $I$ in the \IIE{} model.
        \item \label{sim:com} If $I$ is a $\ComInst$ fence, then $I$ must have been mapped to the next instruction to execute in some processor (we will prove this), and we perform the \wmmComRule{} operation to execute $I$ in the \IIE{} model.
        \item \label{sim:pop-ld} $I$ must be a load in this case.
        If $I$ has been mapped, then it must be mapped to the next instruction to execute in some processor in the $\IIE$ model (we will prove this), and we perform the \wmmLdRule{} operation to execute $I$ in the $\IIE$ model.
        Otherwise, we just add $I$ into set $Z$.
    \end{enumerate}
\end{enumerate}
For proof purposes, we introduce the following ghost states to the $\IIE$ model:
\begin{itemize}
    \item Field $\SrcSt$ in \regfile{} memory: For each address $a$, we add state $m[a].\SrcSt$ to record the store that writes the current memory value.
    \item Fields $\SrcSt$ in invalidation buffer: For each stale value $\langle a, v\rangle$ in an invalidation buffer, we add state $v.\SrcSt$ to denote the store of this stale value.
\end{itemize}
These two fields are set when a \wmmDeqSbRule{} operation is performed.
Assume the \wmmDeqSbRule{} operation dequeues a store $S$ for address $a$.
In this case, we update $m[a].\SrcSt$ to be $S$.
Let $S_0$ be the original $m[a].\SrcSt$ before this operation is performed.
Then for each new stale value $\langle a, v\rangle$ inserted into any invalidation buffer, we set $v.\SrcSt = S_0$.
It is obvious that memory value $m[a]$ is equal to the value of $m[a].\SrcSt$, and stale value $v$ is equal to the value of $v.\SrcSt$.

For proof purposes, we define a function $\OverwriteSt$.
For each store $S$ in $\MemOrd$, $\OverwriteSt(S)$ returns the store for the same address such that
\begin{itemize}
    \item $S\MemOrd\OverwriteSt(S)$, and
    \item there is no store $S'$ for the same address such that $S\MemOrd S'\MemOrd\OverwriteSt(S)$.
\end{itemize}
In other words, $\OverwriteSt(S)$ returns the store that overwrites $S$ in $\MemOrd$.
($\OverwriteSt(S)$ does not exist if $S$ is the last store for its address in $\MemOrd$.)

Also for proof purposes, at each time in the algorithm, we use $V_i$ to represent the set of every store $S$ in $\MemOrd$ that satisfies all the following requirements:
\begin{enumerate}
    \item \label{req:ib:sb} The store buffer of processor $i$ does not contain the address of $S$.
    \item \label{req:ib:ovw} $\OverwriteSt(S)$ exists and $\OverwriteSt(S)$ has been popped from $Q$.
    \item \label{req:ib:rec} For each $\RecInst$ fence $F$ that has been executed by processor $i$ in the $\IIE$ model, $F \MemOrd \OverwriteSt(S)$.
    \item \label{req:ib:st} For each store $S'$ for the same address that has been executed by processor $i$ in the $\IIE$ model, $S' \MemOrd \OverwriteSt(S)$.
    \item \label{req:ib:ld} For each load $L$ for the same address that has been executed by processor $i$ in the $\IIE$ model, if store $S'\ReadFrom L$ in the axiomatic relations, then $S'\MemOrd \OverwriteSt(S)$.
\end{enumerate}

With the above definitions and new states, we introduce the invariants of the algorithm.
After each step of the algorithm, we have the following invariants for the states of the $\IIE$ model, $Z$ and $Q$:
\begin{enumerate}
    \item \label{inv:comp:po} For each processor $i$, the execution order of all executed instructions in processor $i$ in the $\IIE$ model is a prefix of the $\ProgOrd$ of processor $i$ in the axiomatic relations.
    \item \label{inv:comp:guard} The predicate of any operation performed in this step is satisfied.
    \item \label{inv:comp:inst-res} If we perform an operation to execute an instruction in the $\IIE$ model in this step, the operation is able to get the same instruction result as that of the corresponding instruction in the axiomatic relations.
    \item \label{inv:comp:addr-data} The instruction type, load/stores address, and store data of every mapped instruction in the $\IIE$ model are the same as those of the corresponding instruction in the axiomatic relations.
    \item \label{inv:comp:ld} All loads that have been executed in the $\IIE$ model are mapped exactly to  all loads in $\MemOrd$ but not in $Q$ or $Z$.
    \item \label{inv:comp:fence} All fences that have been executed in processor $i$ are mapped exactly to all fences in $\MemOrd$ but not in $Q$.
    \item \label{inv:comp:st-deq} All stores that have been executed and dequeued from the store buffers in the $\IIE$ model are mapped exactly to all stores in $\MemOrd$ but not in $Q$.
    \item \label{inv:comp:mem} For each address $a$, $m[a].\SrcSt$ in the $\IIE$ model is mapped to the youngest store for $a$, which has been popped from $Q$, in $\MemOrd$.
    \item \label{inv:comp:sb} For each processor $i$, the store buffer of processor $i$ contains exactly every store that has been executed in the $\IIE$ model but still in $Q$.
    \item \label{inv:comp:sb-ord} For any two stores $S_1$ and $S_2$ for the same address in the store buffer of any processor $i$ in the $\IIE$ model, if $S_1$ is older than $S_2$ in the store buffer, then $S_1 \ProgOrd S_2$.
    \item \label{inv:comp:sb-ib} For any processor $i$ and any address $a$, address $a$ cannot be present in the store buffer and invalidation buffer of processor $i$ at the same time.
    \item \label{inv:comp:ib} For any processor $i$, for each store $S$ in $V_i$, the invalidation buffer of processor $i$ contains an entry whose $\SrcSt$ field is mapped to $S$.
    \item \label{inv:comp:stale} For any stale value $\langle a, v\rangle$ in any invalidation buffer, $v.\SrcSt$ has been mapped to a store in $\MemOrd$, and $\OverwriteSt(v.\SrcSt)$ exists, and $\OverwriteSt(v.\SrcSt)$ is not in $Q$.
    \item \label{inv:comp:ib-ord} For any two stale values $v_1$ and $v_2$ for the same address in the invalidation buffer of any processor $i$ in the $\IIE$ model, if $v_1$ is older than $v_2$ in the invalidation buffer, then $v_1.\SrcSt \MemOrd v_2.\SrcSt$.
\end{enumerate}
These invariants guarantee that the algorithm will operate the $\IIE$ model to produce the same program behavior as the axiomatic model.
We now prove inductively that all invariants hold after each step of the algorithm, i.e., we assume all invariants hold before the step.
In case a state is changed in this step, we use superscript 0 to denote the state before this step (e.g., $Q^0$) and use superscript 1 to denote the state after this step (e.g., $Q^1$).
We consider which action is performed in this step.
\begin{itemize}
    \item Action~\ref{sim:nm}: We perform a \wmmNmRule{} operation that executes a non-memory instruction in the $\IIE$ model.
    All the invariants still hold after this step.
    
    \item Action~\ref{sim:ex-st}: We perform a \wmmStRule{} operation that executes a store $S$ for address $a$ in processor $i$ in the $\IIE$ model.
    We consider each invariant.
    \begin{itemize}
        \item Invariants~\ref{inv:comp:po}, \ref{inv:comp:guard}, \ref{inv:comp:addr-data}: These invariants obviously hold.
        \item Invariants~\ref{inv:comp:inst-res}, \ref{inv:comp:ld}, \ref{inv:comp:fence}, \ref{inv:comp:st-deq}, \ref{inv:comp:mem}: These are not affected.
        \item Invariant~\ref{inv:comp:sb}: Note that $S$ is mapped before this step.
        Since $S$ cannot be dequeued from the store buffer before this step, invariant~\ref{inv:comp:st-deq} says that $S$ is still in $Q$.
        Thus, this invariant holds.
        \item Invariant~\ref{inv:comp:sb-ord}: Since $S$ is the youngest store in store buffer and invariant~\ref{inv:comp:po} holds after this step, this invariant also holds.
        \item Invariant~\ref{inv:comp:sb-ib}: Since the \wmmStRule{} operation removes all stale values for $a$ from the invalidation buffer of processor $i$, this invariant still holds.
        \item Invariant~\ref{inv:comp:ib}: For any processor $j$ ($j\neq i$), the action in this step cannot change $V_j$ or the invalidation buffer of processor $j$.
        We only need to consider processor $i$.
        The action in this step cannot introduce any new store into $V_i$, i.e., $V_i^1 \subseteq V_i^0$.
        Also notice that $V_i^1$ does not contain any store for $a$ due to requirement~\ref{req:ib:sb}.
        Since the action in this step only removes values for address $a$ from the invalidation buffer of processor $i$, this invariant still holds for $i$.
        \item Invariants~\ref{inv:comp:stale}, \ref{inv:comp:ib-ord}: These still hold, because we can only remove values from  the invalidation buffer in this step.
    \end{itemize}
    
    \item Action~\ref{sim:ld-z}: We perform a \wmmLdRule{} operation that executes a load $L$ in $Z$.
    (Note that $L$ has been popped from $Q$ before.)
    We assume $L$ is in processor $i$ (both the axiomatic relations and the $\IIE$ model agree on this because of the way we create mappings).
    We also assume that $L$ has address $a$ in the axiomatic relations, and that store $S\ReadFrom L$ in the axiomatic relations.
    According to invariant~\ref{inv:comp:addr-data}, $L$ also has load address $a$ in the $\IIE$ model.
    We first consider several simple invariants:
    \begin{itemize}
        \item Invariants~\ref{inv:comp:po}, \ref{inv:comp:guard}, \ref{inv:comp:ld}: These invariants obviously hold.
        \item Invariants~\ref{inv:comp:fence}, \ref{inv:comp:st-deq}, \ref{inv:comp:mem}, \ref{inv:comp:sb}, \ref{inv:comp:sb-ord}: These are not affected.
        \item Invariant~\ref{inv:comp:sb-ib}: Since the \wmmLdRule{} operation does not change store buffers and can only remove values from the invalidation buffers, this invariant still holds.
        \item Invariants~\ref{inv:comp:stale}, \ref{inv:comp:ib-ord}: These still hold, because we can only remove values from  the invalidation buffer in this step.
    \end{itemize}
    We now consider the remaining invariants, i.e., \ref{inv:comp:inst-res}, \ref{inv:comp:addr-data} and \ref{inv:comp:ib}, according to the current state of $Q$ (note that $Q$ is not changed in this step):
    \begin{enumerate}
        \item $S$ is in $Q$: We show that the \wmmLdRule{} operation can read the value of $S$ from the store buffer of processor $i$ in the $\IIE$ model.
        We first show that $S$ is in the store buffer of processor $i$.
        Since $L$ is not in $Q$, we have $L\MemOrd S$.
        According to the \AxiLdVal{} axiom, we know $S\ProgOrd L$, so $S$ must have been executed.
        Since $S$ is in $Q$, invariant~\ref{inv:comp:st-deq} says that $S$ cannot be dequeued from the store buffer, i.e., $S$ is in the store buffer of processor $i$.
        
        Now we prove that $S$ is the youngest store for $a$ in the store buffer of processor $i$ by contradiction, i.e., we assume there is another store $S'$ for $a$ which is in the store buffer of processor $i$ and is younger than $S$.
        Invariant~\ref{inv:comp:sb-ord} says that $S \ProgOrd S'$.
        Since $S$ and $S'$ are stores for the same address, the \AxiInstOrd{} axiom says that $S\MemOrd S'$.
        Since $S'$ is in the store buffer, it is executed before $L$.
        According to invariant~\ref{inv:comp:po}, $S'\ProgOrd L$.
        Then $S\ReadFrom L$ contradicts with the \AxiLdVal{} axiom.
        
        Now we can prove the invariants:
        \begin{itemize}
            \item Invariant~\ref{inv:comp:inst-res}: This holds because the \wmmLdRule{} operation reads $S$ from the store buffer.
            \item Invariant~\ref{inv:comp:addr-data}: This holds because invariant~\ref{inv:comp:inst-res} holds after this step.
            \item Invariant~\ref{inv:comp:ib}: The execution of $L$ in this step cannot introduce new stores into $V_j$ for any $j$, i.e., $V_j^1 \subseteq V_j^0$.
            Since there is no change to any invalidation buffer when \wmmLdRule{} reads from the store buffer, this invariant still holds.
        \end{itemize}
        
        \item $S$ is not in $Q$ but $\OverwriteSt(S)$ is in $Q$:
        We show that the \wmmLdRule{} operation can read the value of $S$ from the \regfile{} memory.
        Since $S$ has been popped from $Q$ while $\OverwriteSt(S)$ is not, $S$ is the youngest store for $a$ in $\MemOrd$ that has been popped from $Q$.
        According to invariant~\ref{inv:comp:mem}, the current $m[a].\SrcSt$ in the $\IIE$ model is $S$.
        To let \wmmLdRule{} read $m[a]$, we only need to show that the store buffer of processor $i$ does not contain any store for $a$.
        We prove by contradiction, i.e., we assume there is a store $S'$ for $a$ in the store buffer of processor $i$.
        According to invariant~\ref{inv:comp:sb}, $S'$ has been executed in the $\IIE$ model, and $S'$ is still in $Q$.
        Thus, we have $S'\ProgOrd L$ (according to invariant~\ref{inv:comp:po}), and $S \MemOrd S'$.
        Then $S\ReadFrom L$ contradicts with the \AxiLdVal{} axiom.
        
        Now we can prove the invariants:
        \begin{itemize}
            \item Invariant~\ref{inv:comp:inst-res}: This holds because the \wmmLdRule{} operation reads $S$ from the \regfile{} memory $m[a]$.
            \item Invariant~\ref{inv:comp:addr-data}: This holds because invariant~\ref{inv:comp:inst-res} holds after this step.
            \item Invariant~\ref{inv:comp:ib}: The execution of $L$ in this step cannot introduce new stores into $V_j$ for any $j$, i.e., $V_j^1 \subseteq V_j^0$.
            Since there is no change to any invalidation buffer of any processor other than $i$, we only need to consider processor $i$.
            The \wmmLdRule{} removes all values for $a$ from the invalidation buffer of processor $i$, so the goal is to show that there is no store for $a$ in $V_i^1$.
            We prove by contradiction, i.e., assume there is store $S'$ for $a$ in $V_i^1$.
            Requirement~\ref{req:ib:ld} for $V_i$ says that $S\MemOrd \OverwriteSt(S')$.
            Since $\OverwriteSt(S)$ is in $Q$, $\OverwriteSt(S')$ is also in $Q$.
            This contradicts with requirement~\ref{req:ib:ovw}.
            Therefore, there is no store for $a$ in $V_i^1$, and this invariant holds.
        \end{itemize}
        
        \item Both $S$ and $\OverwriteSt(S)$ are not in $Q$:
        We show that the \wmmLdRule{} operation can read the value of $S$ from the invalidation buffer of processor $i$.
        That is, we need to show $S\in V_i^0$.
        We now prove that $S$ satisfies all the requirements for $V_i^0$:
        \begin{itemize}
            \item Requirement~\ref{req:ib:sb}: We prove by contradiction, i.e., we assume there is store $S'$ for $a$ in the store buffer of processor $i$.
            Invariant~\ref{inv:comp:sb} says that $S'$ has been executed but not in $Q$.
            Then we have $S'\ProgOrd L$ (invariant~\ref{inv:comp:po}) and $S \MemOrd S'$.
            Then $S\ReadFrom L$ contradicts with the \AxiLdVal{} axiom.
            \item Requirement~\ref{req:ib:ovw}: This satisfied because we assume $\OverwriteSt(S)$ is not in $Q$.
            \item Requirement~\ref{req:ib:rec}: We prove by contradiction, i.e., we assume that $\RecInst$ fence $F$ has been executed by processor $i$, and $\OverwriteSt(S) \MemOrd F$.
            Since $F$ is executed before $L$, invariant~\ref{inv:comp:po} says that $F\ProgOrd L$.
            Since $\orderFunc(F, L)$, the \AxiInstOrd{} axiom says that $F\MemOrd L$.
            Now we have $S\MemOrd \OverwriteSt(S) \MemOrd F\MemOrd L$.
            Thus, $S\ReadFrom L$ contradicts with the \AxiLdVal{} axiom.
            \item Requirement~\ref{req:ib:st}: We prove by contradiction, i.e., we assume that store $S'$ for $a$ has been executed by processor $i$, and either $S'=\OverwriteSt(S)$ or $\OverwriteSt(S) \MemOrd S'$.
            According to the definition of $\OverwriteSt$, we have $S\MemOrd S'$.
            Since $S'$ has been executed, invariant~\ref{inv:comp:po} says that $S'\ProgOrd L$.
            Then $S\ReadFrom L$ contradicts with the \AxiLdVal{} axiom.
            \item Requirement~\ref{req:ib:ld}: We prove by contradiction, i.e., we assume that store $S'$ and load $L'$ are both for address $a$, $L'$ has been executed by processor $i$, $S'\ReadFrom L'$, and either $S'=\OverwriteSt(S)$ or $\OverwriteSt(S) \MemOrd S'$.
            According to the definition of $\OverwriteSt$, we have $S\MemOrd S'$.
            Since $L'$ has been executed, invariant~\ref{inv:comp:po} says that $L'\ProgOrd L$.
            Since $\orderFunc(L', L)$, the \AxiInstOrd{} axiom says that $L'\MemOrd L$.
            Since $S'\ReadFrom L$, we have $S'\ProgOrd L'$ or $S'\MemOrd L'$.
            Since $L'\ProgOrd L$ and $L'\MemOrd L$, we have $S'\ProgOrd L$ or $S'\MemOrd L$.
            Since $S\MemOrd S'$, $S\ReadFrom L$ contradicts with the \AxiLdVal{} axiom.
        \end{itemize}
        
        Now we can prove the invariants:
        \begin{itemize}
            \item Invariant~\ref{inv:comp:inst-res}: This holds because the \wmmLdRule{} operation reads $S$ from the invalidation buffer of processor $i$.
            \item Invariant~\ref{inv:comp:addr-data}: This holds because invariant~\ref{inv:comp:inst-res} holds after this step.
            \item Invariant~\ref{inv:comp:ib}: The execution of $L$ in this step cannot introduce new stores into $V_j$ for any $j$, i.e., $V_j^1 \subseteq V_j^0$.
            Since there is no change to any invalidation buffer of any processor other than $i$, we only need to consider processor $i$.
            Assume the invalidation buffer entry read by the \wmmLdRule{} operation is $\langle a, v \rangle$, and $v.\SrcSt = S$.
            The \wmmLdRule{} rule removes any stale value $\langle a, v'\rangle$ that is older than $\langle a, v \rangle$ from the invalidation buffer of processor $i$.
            The goal is to show that $v'.\SrcSt$ cannot be in $V_i^1$.
            We prove by contradiction, i.e., we assume that $v'.\SrcSt \in V_i^1$.
            Since $L$ has been executed after this step, requirement~\ref{req:ib:ld} says that $S \MemOrd \OverwriteSt(v'.\SrcSt)$.
            Since $v'$ is older than $v$ in the invalidation buffer before this step, invariant~\ref{inv:comp:ib-ord} says that $v'.\SrcSt \MemOrd v.\SrcSt = S$.
            The above two statements contradict with each other.
            Therefore, this invariant still holds.
        \end{itemize}
    \end{enumerate}
    
    \item Action~\ref{sim:deq-st}: We pop a store $S$ from the head of $Q$, and we perform a \wmmDeqSbRule{} operation to dequeue $S$ from the store buffer.
    Assume that $S$ is for address $a$, and in processor $i$ in the axiomatic relations.
    We first prove that $S$ has been mapped before this step.
    We prove by contradiction, i.e., we assume $S$ has not been mapped to any instruction in the $\IIE$ model before this step.
    Consider the state right before this step.
    Let $I$ be the next instruction to execute in processor $i$ in the $\IIE$ model.
    We know $I$ is mapped and $I\ProgOrd S$.
    The condition for performing action~\ref{sim:deq-st} in this step says that $I$ can only be a fence or load, and we have $\orderFunc(I, S)$.
    According to the \AxiInstOrd{} axiom, $I\MemOrd S$, so $I$ has been popped from $Q^0$.
    \begin{enumerate}
        \item If $I$ is a fence, since $I$ is in $\MemOrd$ but not in $Q^0$, invariant~\ref{inv:comp:fence} says that $I$ must be executed, contradicting our assumption that $I$ is the next instruction to execute.
        \item If $I$ is a load, since $I$ is not executed, and $I$ is in $\MemOrd$, and $I$ is not in $Q^0$, invariant~\ref{inv:comp:ld} says that $I$ must be in $Z^0$.
        Then this algorithm step should use action~\ref{sim:ld-z} instead of action~\ref{sim:deq-st}.
    \end{enumerate}
    Due to the contradictions, we know $S$ must have been mapped.
    Note that the next instruction to execute in processor $i$ cannot be store, because otherwise this step will use action~\ref{sim:ex-st}.
    According to invariant~\ref{inv:comp:addr-data}, $S$ cannot be mapped to the next instruction to execution in processor $i$.
    Therefore $S$ must have been executed in processor $i$ in the $\IIE$ model before this step.
    
    Also according to invariant~\ref{inv:comp:addr-data}, the address and data of $S$ in the $\IIE$ model are the same as those in the axiomatic relations.
    Now we consider each invariant.
    \begin{itemize}
        \item Invariants~\ref{inv:comp:po}, \ref{inv:comp:addr-data}, \ref{inv:comp:st-deq}, \ref{inv:comp:sb}, \ref{inv:comp:sb-ord}: These invariants obviously hold.
        \item Invariants~\ref{inv:comp:inst-res}, \ref{inv:comp:ld}, \ref{inv:comp:fence}: These are not affected.
        
        \item Invariant~\ref{inv:comp:guard}: We prove by contradiction, i.e., we assume there is a store $S'$ for $a$ younger than $S'$ in the store buffer of processor $i$ (before this step).
        According to invariant~\ref{inv:comp:sb}, $S'$ is in $Q$.
        Since $S$ is the head of $Q^0$, $S\MemOrd S'$.
        According to invariant~\ref{inv:comp:sb-ord}, $S'\ProgOrd S$.
        Since $\orderFunc(S', S)$, $S'\MemOrd S$, contradicting with previous statement.
        Thus, the predicate of the \wmmDeqSbRule{} operation is satisfied, and the invariant holds.
        
        \item Invariant~\ref{inv:comp:mem}: $S$ is the youngest instruction in $\MemOrd$ that has been popped from $Q$, and $m[a].\SrcSt$ is updated to $S$.
        Thus, this invariant still holds.
        
        \item Invariant~\ref{inv:comp:sb-ib}: Since the  \wmmDeqSbRule{} operation will not insert the stale value into an invalidation buffer of a processor if the store buffer of that processor contains the same address, this invariant still holds.
        
        \item Invariant~\ref{inv:comp:ib}: For any processor $j$, the action this step will not remove stores from $V_j$ but may introduce new stores to $V_j$, i.e., $V_j^0 \subseteq V_j^1$.
        We consider the following two types of processors.
        \begin{enumerate}
            \item Processor $i$: We show that $V_i^0 = V_i^1$.
            We prove by contradiction, i.e., we assume there is store $S'$ such that $S'\in V_i^1$ but $S'\notin V_i^0$.
            Since $S'$ satisfies requirements~\ref{req:ib:rec}, \ref{req:ib:st}, \ref{req:ib:ld} after this step, it also satisfies these three requirements before this step.
            Then $S'$ must fail to meet at least one of requirements~\ref{req:ib:sb} and \ref{req:ib:ovw} before this step.
            \begin{enumerate}
                \item If $S'$ does not meet requirement~\ref{req:ib:sb} before this step, then $S'.\Addr$ is in the store buffer of processor $i$ before this step and $S'.\Addr$ is not in this store buffer after this step.
                Thus, $S'.\Addr$ must be $a$.
                Since $S'$ meets requirement~\ref{req:ib:st} before this step and $S$ has been executed by processor $i$ before this step, we know $S\MemOrd\OverwriteSt(S')$.
                Since $S'$ meets requirement~\ref{req:ib:ovw} after this step, $\OverwriteSt(S')$ is not in $Q^1$.
                Since $Q^1$ is derived by popping the oldest store form $Q^0$, we know $S$ is not in $Q^0$.
                Since $S$ is in the store buffer before this step, this contradicts invariant~\ref{inv:comp:sb}.
                Therefore this case is impossible.
                \item If $S'$ does not meet requirement~\ref{req:ib:ovw} before this step, then $\OverwriteSt(S')$ is in $Q^0$ but not in $Q^1$.
                Then $\OverwriteSt(S') = S$.
                Since $S$ has been executed by processor $i$, $S'$ will fail to meet requirement~\ref{req:ib:st} after this step.
                This contradicts with $S'\in V_i^1$, so this case is impossible either.
            \end{enumerate}
            Now we have proved that $V_i^0 = V_i^1$.
            Since the \wmmDeqSbRule{} operation does not change the invalidation buffer of processor $i$, this invariant holds for processor $i$.
            
            \item Processor $j$ ($\neq i$):
            We consider any store $S'$ such that $S'\in V_j^1$ but $S'\notin V_j^0$.
            Since $S'$ satisfies requirements~\ref{req:ib:sb}, \ref{req:ib:rec}, \ref{req:ib:st}, \ref{req:ib:ld} after this step, it also satisfies these four requirements before this step.
            Then $S'$ must fail to meet requirement requirement~\ref{req:ib:ovw} before this step, i.e., $\OverwriteSt(S')$ is in $Q^0$ but not in $Q^1$.
            Then $\OverwriteSt(S') = S$.
            According to invariant~\ref{inv:comp:mem}, we know $S'$ is $m[a].\SrcSt$.
            Consider the following two cases.
            \begin{enumerate}
                \item The store buffer of processor $j$ contains address $a$:
                Since $S'$ cannot meet requirement~\ref{req:ib:sb}, $V_j^1 = V_j^0$.
                Since \wmmDeqSbRule{} cannot remove any value from the invalidation buffer of processor $j$, this invariant holds.
                \item The store buffer of processor $j$ does not contain address $a$:
                In this case,  the \wmmDeqSbRule{} operation will insert stale value $\langle a, m[a]^0\rangle$ into the invalidation buffer of processor $j$, so the invariant still holds.
            \end{enumerate}
        \end{enumerate}        
        
        \item Invariant~\ref{inv:comp:stale}: The $\SrcSt$ field of all the newly inserted stale values in this step are equal to $m[a].\SrcSt^0$.
        According to invariant~\ref{inv:comp:mem}, $m[a].\SrcSt^0$ is the youngest store for $a$ in $\MemOrd$ that is not in $Q^0$.
        Since $S$ is the head of $Q^0$, we know $\OverwriteSt(m[a].\SrcSt^0) = S$.
        Since $S$ is not in $Q^1$, this invariant still holds.
        
        \item Invariant~\ref{inv:comp:ib-ord}: Assume $\langle a,v\rangle$ is the new stale value inserted into the invalidation buffer of processor $j$ in this step.
        We need to show that for any stale value $\langle a,v'\rangle$ that is in this invalidation buffer before this step, $v'.\SrcSt\MemOrd v.\SrcSt$.
        According to invariant~\ref{inv:comp:stale}, $\OverwriteSt(v'.\SrcSt)$ is not in $Q^0$.
        Since $v.\SrcSt = m[a].\SrcSt^0$, according to invariant~\ref{inv:comp:mem}, $v.\SrcSt$ is the youngest store for $a$ in $\MemOrd$ that is not in $Q^0$.
        Therefore $v'.\SrcSt\MemOrd v.\SrcSt$.
    \end{itemize}
    
    \item Action~\ref{sim:rec}: We pop a $\RecInst$ fence $F$ from $Q$, and perform a \wmmRecRule{} operation to execute it in the $\IIE$ model.
    Assume $F$ is in processor $i$ in the axiomatic relations.
    We first prove that $F$ has been mapped before this step. 
    We prove by contradiction, i.e., we assume $F$ is not mapped before this step.
    Consider the state right before this step.
    Let $I$ be the next instruction to execute in processor $i$ in the $\IIE$ model.
    We know $I$ is mapped and $I\ProgOrd F$.
    The condition for performing action~\ref{sim:rec} in this step says that $I$ can only be a fence or load, so we have $\orderFunc(I, F)$.
    According to the \AxiInstOrd{} axiom, $I\MemOrd F$, so $I$ has been popped from $Q^0$.
    \begin{enumerate}
        \item If $I$ is a fence, since $I$ is in $\MemOrd$ but not in $Q^0$, invariant~\ref{inv:comp:fence} says that $I$ must be executed, contradicting our assumption that $I$ is the next instruction to execute.
        \item If $I$ is a load, since $I$ is not executed, and $I$ is in $\MemOrd$, and $I_k$ is not in $Q^0$, invariant~\ref{inv:comp:ld} says that $I$ must be in $Z^0$.
        Then this algorithm step should use action~\ref{sim:ld-z} instead of action~\ref{sim:deq-st}.
    \end{enumerate}
    Due to the contradictions, we know $F$ must have been mapped before this step.
    According to invariant~\ref{inv:comp:fence}, since $F$ is in $Q^0$, $F$ must have not been executed in the $\IIE$ model.
    Thus, $F$ is mapped to the next instruction to execute in processor $i$ in the $\IIE$ model.
    
    Now we consider each invariant:
    \begin{itemize}
        \item Invariants~\ref{inv:comp:po}, \ref{inv:comp:guard}, \ref{inv:comp:addr-data}, \ref{inv:comp:fence}: These obviously hold.
        \item Invariants~\ref{inv:comp:inst-res}, \ref{inv:comp:ld}, \ref{inv:comp:st-deq}, \ref{inv:comp:mem}, \ref{inv:comp:sb}, \ref{inv:comp:sb-ord}: These are not affected.
        \item Invariants~\ref{inv:comp:sb-ib}, \ref{inv:comp:stale}, \ref{inv:comp:ib-ord}: These invariants hold, because the invalidation buffer of any processor $j$ ($\neq i$) is not changed, and the invalidation buffer of processor $i$ is empty after this step.
        \item Invariant~\ref{inv:comp:ib}: For any processor $j$ ($\neq i$), $V_j^1 = V_j^0$ and the invalidation buffer of processor $j$ is not changed in this step.
        Thus, this invariant holds for any processor $j$ ($\neq i$).
        We now consider processor $i$.
        The invalidation buffer of processor $i$ is empty after this step, so we need to show that $V_i^1$ is empty.
        We prove by contradiction, i.e., we assume there is a store $S \in V_i^1$.
        Since $F$ has been executed in processor $i$ after this step, requirement~\ref{req:ib:rec} says that $F\MemOrd\OverwriteSt(S)$.
        Since $F$ is the head of $Q^0$, $\OverwriteSt(S)$ must be in $Q^1$.
        Then $S$ fails to meet requirement~\ref{req:ib:ovw} after this step, contradicting with $S \in V_i^1$.
        Therefore $V_i^1$ is empty, and this invariant also holds for processor $i$.
    \end{itemize}
    
    \item Action~\ref{sim:com}: We pop a $\ComInst$ fence $F$ from $Q$, and perform a \wmmComRule{} operation to execute it in the $\IIE$ model.
    Assume $F$ is in processor $i$ in the axiomatic relations.
    Using the same argument as in the previous case, we can prove that $F$ is mapped to the next instruction to execute in processor $i$ in the $\IIE$ model (before this step).
    Now we consider each invariant:
    \begin{itemize}
        \item Invariants~\ref{inv:comp:po}, \ref{inv:comp:addr-data}, \ref{inv:comp:fence}: These obviously hold.
        \item Invariants~\ref{inv:comp:inst-res}, \ref{inv:comp:ld}, \ref{inv:comp:st-deq}, \ref{inv:comp:mem}, \ref{inv:comp:sb}, \ref{inv:comp:sb-ord}, \ref{inv:comp:sb-ib}, \ref{inv:comp:ib}, \ref{inv:comp:stale}, \ref{inv:comp:ib-ord}: These are not affected.
        \item Invariants~\ref{inv:comp:guard}: We prove by contradiction, i.e., we assume there is store $S$ in the store buffer of processor $i$ before this step.
        According to invariant~\ref{inv:comp:sb}, $S$ has been executed in processor $i$ and is in $Q^0$.
        Thus, we have $S\ProgOrd F$.
        Since $\orderFunc(S, F)$, the \AxiLdVal{} axiom says that $S\MemOrd F$.
        Then $F$ is not the head of $Q^0$, contradicting with the fact that we pop $F$ from the head of $Q^0$.
    \end{itemize}
    
    \item Action~\ref{sim:pop-ld}: We pop a load $L$ from $Q$.
    Assume that $L$ is for address $a$ and is in processor $i$ in the axiomatic relations.
    If we add $L$ to $Z$, then all invariants obviously hold.
    We only need to consider the case that we perform a \wmmLdRule{} operation to execute $L$ in the $\IIE$ model.
    In this case, $L$ is mapped before this step.
    Since $L$ is in $Q^0$, according to invariant~\ref{inv:comp:ld}, $L$ must be mapped to an unexecuted instruction in the $\IIE$ model.
    That is, $L$ is mapped to the next instruction to execute in processor $i$ in the $\IIE$ model.
    Invariant~\ref{inv:comp:addr-data} ensures that $L$ has the same load address in the $\IIE$ model.
    We first consider several simple invariants:
    \begin{itemize}
        \item Invariants~\ref{inv:comp:po}, \ref{inv:comp:guard}, \ref{inv:comp:ld}: These invariants obviously hold.
        \item Invariants~\ref{inv:comp:fence}, \ref{inv:comp:st-deq}, \ref{inv:comp:mem}, \ref{inv:comp:sb}, \ref{inv:comp:sb-ord}: These are not affected.
        \item Invariant~\ref{inv:comp:sb-ib}: Since the \wmmLdRule{} operation does not change store buffers and can only remove values from the invalidation buffers, this invariant still holds.
        \item Invariants~\ref{inv:comp:stale}, \ref{inv:comp:ib-ord}: These still hold, because we can only remove values from  the invalidation buffer in this step.
    \end{itemize}
    Assume store $S\ReadFrom L$ in the axiomatic relations.
    We prove the remaining invariants (i.e., \ref{inv:comp:inst-res}, \ref{inv:comp:addr-data} and \ref{inv:comp:ib}) according to the position of $S$ in $\MemOrd$.
    \begin{enumerate}
        \item $L\MemOrd S$: We show that the \wmmLdRule{} can read $S$ from the store buffer of processor $i$ in the $\IIE$ model.
        The \AxiLdVal{} axiom says that $S\ProgOrd L$.
        Then $S$ must have been executed in processor $i$ in the $\IIE$ model according to invariant~\ref{inv:comp:po}.
        Since $S$ is in $Q^0$, invariant~\ref{inv:comp:sb} ensures that $S$ is in the store buffer of processor $i$ before this step.
        
        To let \wmmLdRule{} read $S$ from the store buffer, we now only need to prove that $S$ is the youngest store for $a$ in the store buffer of processor $i$.
        We prove by contradiction, i.e., we assume there is another store $S'$ for $a$ which is in the store buffer of processor $i$ and is younger than $S$.
        Invariant~\ref{inv:comp:sb-ord} says that $S \ProgOrd S'$.
        Since $S$ and $S'$ are stores for the same address, the \AxiInstOrd{} axiom says that $S\MemOrd S'$.
        Since $S'$ is in the store buffer, it is executed before $L$.
        According to invariant~\ref{inv:comp:po}, $S'\ProgOrd L$.
        Then $S\ReadFrom L$ contradicts with the \AxiLdVal{} axiom.
        
        Now we can prove the invariants:
        \begin{itemize}
            \item Invariant~\ref{inv:comp:inst-res}: This holds because the \wmmLdRule{} operation reads $S$ from the store buffer.
            \item Invariant~\ref{inv:comp:addr-data}: This holds because invariant~\ref{inv:comp:inst-res} holds after this step.
            \item Invariant~\ref{inv:comp:ib}: The execution of $L$ in this step cannot introduce new stores into $V_j$ for any $j$, i.e., $V_j^1 \subseteq V_j^0$.
            Since there is no change to any invalidation buffer when \wmmLdRule{} reads from the store buffer, this invariant still holds.
        \end{itemize}
        
        \item $S \MemOrd L$:
        We show that the \wmmLdRule{} operation can read the value of $S$ from the \regfile{} memory.
        Since $L$ is the head of $Q^0$, $S$ is not in $Q^0$.
        According to the \AxiLdVal{} axiom, there cannot be any store for $a$ between $S$ and $L$ in $\MemOrd$.
        Thus, $S$ is the youngest store for $a$ in $\MemOrd$ that has been popped from $Q$.
        According to invariant~\ref{inv:comp:mem}, the current $m[a].\SrcSt$ in the $\IIE$ model is $S$.
        
        To let \wmmLdRule{} read $m[a]$, we only need to show that the store buffer of processor $i$ does not contain any store for $a$.
        We prove by contradiction, i.e., we assume there is a store $S'$ for $a$ in the store buffer of processor $i$.
        According to invariant~\ref{inv:comp:sb}, $S'$ has executed in the $\IIE$ model, and $S'$ is in $Q^0$.
        Thus, we have $S'\ProgOrd L$ (according to invariant~\ref{inv:comp:po}), and $S \MemOrd S'$.
        Then $S\ReadFrom L$ contradicts with the \AxiLdVal{} axiom.
        
        Now we can prove the invariants:
        \begin{itemize}
            \item Invariant~\ref{inv:comp:inst-res}: This holds because the \wmmLdRule{} operation reads $S$ from the \regfile{} memory $m[a]$.
            \item Invariant~\ref{inv:comp:addr-data}: This holds because invariant~\ref{inv:comp:inst-res} holds after this step.
            \item Invariant~\ref{inv:comp:ib}: The execution of $L$ in this step cannot introduce new stores into $V_j$ for any $j$, i.e., $V_j^1 \subseteq V_j^0$.
            Since there is no change to any invalidation buffer of any processor other than $i$, we only need to consider processor $i$.
            The \wmmLdRule{} removes all values for $a$ from the invalidation buffer of processor $i$, so the goal is to show that there is no store for $a$ in $V_i^1$.
            We prove by contradiction, i.e., assume there is store $S'$ for $a$ in $V_i^1$.
            Requirement~\ref{req:ib:ld} for $V_i$ says that $S\MemOrd \OverwriteSt(S')$.
            Since $S$ is the youngest store for $a$ that is in $\MemOrd$ but not $Q^0$, $\OverwriteSt(S')$ must be in $Q^0$.
            This contradicts with requirement~\ref{req:ib:ovw}.
            Therefore, there is no store for $a$ in $V_i^1$, and this invariant holds.
        \end{itemize}
    \end{enumerate}
\end{itemize}
\end{proof}

By combining Theorems~\ref{thm:wmm:i2e<axi} and \ref{thm:wmm:axi<i2e}, we prove the equivalence between the \IIE{} model and the axiomatic model of WMM.
\begin{theorem}[Equivalence]
    WMM \IIE{} model $\equiv$ WMM axiomatic model.
\end{theorem}

\bibliographystyle{IEEEtran}
\bibliography{ref}

\end{document}